\documentclass[letterpaper, 10 pt, journal, twoside]{IEEEtran}

\usepackage{graphicx}
\usepackage{amsmath}
\usepackage{amssymb}
\usepackage{mathtools}

\usepackage{amsfonts}
\usepackage{caption}
\usepackage{subcaption}
\usepackage{tikz}
\usepackage{textcomp}
\usepackage{hyperref}
\usepackage{lipsum}

\newtheorem{problem}{Problem}
\newtheorem{definition}{Definition}
\newtheorem{assumption}{Assumption}
\newtheorem{theorem}{Theorem}
\newtheorem{lemma}{Lemma}
\newtheorem{remark}{Remark}

\newtheorem{corollary}{Corollary}
\newenvironment{proof}{\begin{trivlist}\item[] {\em Proof: }}{\hfill $\Box$ \end{trivlist}}

\newcommand\copyrighttext{%
\footnotesize \textcopyright 2012 IEEE. Personal use of this material is permitted.
Permission from IEEE must be obtained for all other uses, in any current or future media, including reprinting/republishing this material for advertising or promotional purposes, creating new collective works, for resale or redistribution to servers or lists, or reuse of any copyrighted component of this work in other works. 
}
\newcommand\copyrightnotice{%
\begin{tikzpicture}[remember picture,overlay]
\node[anchor=south,yshift=10pt] at (current page.south) {\fbox{\parbox{\dimexpr\textwidth-\fboxsep-\fboxrule\relax}{\copyrighttext}}};
\end{tikzpicture}%
}

\begin{document}

\title{PDE-based Dynamic Density Estimation for Large-scale Agent Systems}
\author{Tongjia Zheng, \IEEEmembership{Student Member, IEEE}, Qing Han, and Hai Lin, \IEEEmembership{Senior Member, IEEE}
\thanks{This work was supported in part by the National Science Foundation under Grant IIS-1724070, and Grant CNS-1830335, and in part by the Army Research Laboratory under Grant W911NF-17-1-0072.}
\thanks{Tongia Zheng and Hai Lin are with the Department of Electrical Engineering, University of Notre Dame, Notre Dame, IN 46556, USA (e-mail: tzheng1@nd.edu, hlin1@nd.edu.). }
\thanks{Qing Han is with the Department of Mathematics, University of Notre Dame, Notre Dame, IN 46556, USA (e-mail: Qing.Han.7@nd.edu.).}
}

\maketitle

\copyrightnotice

\thispagestyle{empty}
\pagestyle{empty}

\begin{abstract}
Large-scale agent systems have foreseeable applications in the near future. Estimating their macroscopic density is critical for many density-based optimization and control tasks, such as sensor deployment and city traffic scheduling.
In this paper, we study the problem of estimating their dynamically varying probability density, given the agents' individual dynamics (which can be nonlinear and time-varying) and their states observed in real-time. 
The density evolution is shown to satisfy a linear partial differential equation uniquely determined by the agents' dynamics.
We present a density filter which takes advantage of the system dynamics to gradually improve its estimation and is scalable to the agents' population.
Specifically, we use kernel density estimators (KDE) to construct a noisy measurement and show that, when the agents' population is large, the measurement noise is approximately ``Gaussian''. 
With this important property, infinite-dimensional Kalman filters are used to design density filters. 
It turns out that the covariance of measurement noise depends on the true density.
This state-dependence makes it necessary to approximate the covariance in the associated operator Riccati equation, rendering the density filter suboptimal. 
The notion of input-to-state stability is used to prove that the performance of the suboptimal density filter remains close to the optimal one.
Simulation results suggest that the proposed density filter is able to quickly recognize the underlying modes of the unknown density and automatically ignore outliers, and is robust to different choices of kernel bandwidth of KDE.
\end{abstract}



\section{Introduction}
\IEEEPARstart{E}{fficiently} estimating the continuously varying probability density of the states of large-scale agent systems is an important step for many density-based optimization and control tasks, such as sensor deployment \cite{foderaro2016distributed} and city traffic scheduling. 
In such scenarios, the agents (such as UAVs) are built by task designers and follow the specified control commands, which means their dynamics are known. 
While density estimation is a fundamental problem extensively studied in statistics, we are particularly interested in the case where the samples (i.e. the agent states) are governed by known dynamics (which can be nonlinear and time-varying).
We study how to take advantage of their dynamics to obtain efficient and convergent density estimates in real time.

In general, density estimation algorithms are classified as parametric and non-parametric. Parametric algorithms assume that the samples are drawn from a known parametric family of distributions with a fixed set of parameters, such as the Gaussian mixture models \cite{silverman1986density}. 
Performance of such estimators rely on the validity of the assumed models, and therefore they are unsuitable for estimating an evolving density. In non-parametric approaches, the data are allowed to speak for themselves in determining the density estimate. As a representative, \textit{kernel density estimation (KDE)} \cite{silverman1986density} has been widely used for estimating the global density in the study of swarm robotic systems \cite{foderaro2012decentralized, krishnan2018distributed, de2018optimal}. 
It is known that the performance of KDE largely depends on a smoothing parameter called the bandwidth \cite{silverman1986density}.
When the samples are stationary, many bandwidth selection techniques have been proposed, such as cross-validation \cite{silverman1986density}, adaptive bandwidth \cite{sain2002multivariate} and the plug-in technique \cite{sheather1991reliable}. 
Such techniques are heuristic in general, since any predefined optimality of bandwidth selection requires certain information of the unknown density. 
They are also not suitable for dynamic estimation which requires the algorithm to be computationally efficient and adapt to the density evolution. 
For dynamic density estimation, many adaptations of KDE have been proposed in the domain of data stream mining \cite{qahtan2016kde}.
In such problems, the dynamics of the density are usually unknown, and therefore little can be claimed about the convergence of the algorithm. 

To estimate the global distribution of large-scale agents with known dynamics, an alternative way is to formulate it as a filtering problem, for which there exists a large body of literature, such as the celebrated Kalman filters and their variants \cite{julier2004unscented, arasaratnam2009cubature}, the more general Bayesian filters and their Monte Carlo approach, also known as particle filters \cite{chen2003bayesian}.
Unfortunately, all these methods are known to suffer from the \textit{curse of dimensionality}, especially for large-scale nonlinear systems.
A potential solution is to use the so-called \textit{consensus filters} \cite{olfati2007distributed, bandyopadhyay2014distributed, battistelli2014kullback}, which obtain local distribution estimates of the agent system using local Kalman or Bayesian filters, and then compute the global distribution by averaging local estimates through some consensus protocols. 
However, few conclusions are made concerning the relationship between the true distribution and the estimated global distribution. 
Moreover, stability analysis is difficult when the agents' dynamics are nonlinear and time-varying.

In summary, considering efficiency and convergence requirements, existing methods are unsuitable for estimating the time-varying density of large-scale agent systems. 
This motivates us to propose a dynamic and scalable density estimation algorithm that can perform online and take advantage of the dynamics to guarantee its convergence. 
Specifically, we show that the agents' density is governed by a linear partial differential equation (PDE), called the Fokker-Planck equation, which is uniquely determined by the agents' dynamics. 
KDE is used to construct a noisy measurement of the density (the state of this PDE). 
The measurement noise is approximately ``Gaussian'' (more precisely, asymptotically Gaussian when the agents' population tends to infinity), so that infinite-dimensional Kalman filters can be used to design density filters. 
It turns out that the covariance of measurement noise depends on the true density, for which approximating the covariance is required and the density filter becomes suboptimal. 
We then use the notion of input-to-state stability to prove the stability of the suboptimal density filter and the associated operator Riccati equation. 

Our contributions contain the following aspects: (i) By using a density filter, we can (largely) circumvent the problem of bandwidth selection; (ii) The density filter takes advantage of the agents' dynamics to improve its density estimation, in the sense of minimizing the covariance of estimation error; 
(iii) The density filter is proved to be convergent and is scalable to the population of agents; 
(iv) All the results hold even if the agents' dynamics are nonlinear and time-varying.

The rest of the paper is organized as follows. Section \ref{section:preliminaries} introduces some preliminaries. Problem formulation is given in Section \ref{section:problem formulation}. Section \ref{section:main results} is our main results, in which we present a density filter and then study its stability/optimality. Section \ref{section:simulation} performs an agent-based simulation to verify the effectiveness of the density filter. Section \ref{section:conclusion} summarizes the contribution and points out future research.

\section{Preliminaries} \label{section:preliminaries}
\subsection{Infinite-Dimensional Kalman Filters} \label{section:Kalman filter}
The Kalman filter is an algorithm that uses the system's model, known control inputs and sequential measurements to form a better state estimate \cite{kalman1961new}. It was later extended to infinite-dimensional systems, represented using linear operators \cite{falb1967infinite, bensoussan1971filtrage}.  
Suppose the signal $x(t)$ and its measurement $y(t)$, both in a Hilbert space, are generated by the stochastic linear differential equations
\begin{align*}
    & dx=A(t) x d t+B(t)q(t) d v, \quad x(t_{0})=x_{0}, \quad E[x_{0}]=0\\
    & dy=C(t) x d t+ r(t)d w, \quad y(t_{0})=y_{0}
\end{align*}
where $dv$ and $dw$ are infinite-dimensional Wiener processes with incremental covariance operators $V$ and $W$, respectively \cite{da2014stochastic}. Assume $\operatorname{Cov}[v(t), w(\tau)]=0$ and $E[\langle v(t), w(\tau)\rangle]=0$ for all $t\neq\tau$. Denote $Q(t)=q(t)Vq^*(t)$ and $R(t)=r(t)Wr^*(t)$. The Kalman filter is an algorithm that minimizes the covariance of the estimation error, given by \cite{falb1967infinite}
$$
d \hat{x}=A(t) \hat{x} d t+L(t)[y-C(t) \hat{x}] d t
$$
where $L(t)=P(t) C^{*}(t) R^{-1}(t)$ is the called the optimal Kalman gain, and $P(t)$ is the solution of the Riccati equation
$$
\dot{P}=A P+P A^{*}-P C^{*} R^{-1} C P+B Q B^{*}
$$
with $P(t_{0})=\operatorname{Cov}[x_{0}, x_{0}]$. 

\subsection{Input-to-state stability}
Input-to-state stability (ISS) is a stability notion widely used to study stability of nonlinear control systems with external inputs \cite{sontag1995characterizations}. Roughly speaking, a control system is ISS if it is asymptotically stable in the absence of external inputs and if its trajectories are bounded by a function of the magnitude of the input. To define the ISS concept for infinite-dimensional systems we need to introduce the following classes of comparison functions \cite{sontag1995characterizations}.
\begin{align*}
    \mathcal{K} &:= \{\gamma :\mathbb{R}_{+} \to \mathbb{R}_{+}|\gamma \text{ is continuous and strictly} \\
    &\quad\quad \text{increasing}, \gamma (0)=0\}\\
    \mathcal{L} &:= \{\gamma  : \mathbb{R}_{+} \to \mathbb{R}_{+} | \gamma \text{ is continuous and strictly } \\
    &\quad\quad \text{decreasing with } \lim_{t\to \infty} \gamma (t) = 0\}  \\
    \mathcal{KL} &:=\{\beta:\mathbb{R}_{+} \times\mathbb{R}_{+} \to\mathbb{R}_{+} |\beta \text{ is continuous}, \beta(\cdot,t)\in \mathcal{K}, \\
    &\quad\quad \beta(r,\cdot)\in \mathcal{L}, \forall t\geq 0, \forall r>0\}.
\end{align*}

\begin{definition} \label{definition:(L)ISS}
(\textbf{ISS} \cite{dashkovskiy2013input}). Consider a control system $\Sigma = (X, U, \phi)$ consisting of normed linear spaces $(X, \|\cdot\|_X)$ and $(U, \|\cdot\|_U)$, called the state space and the input space, endowed with the norms $\|\cdot\|_X$ and $\|\cdot\|_U$ respectively, and a transition map $\phi:\mathbb{R}_{+}\times X \times U \to X$. The system is said to be ISS if there exist $\beta\in \mathcal{KL}$ and $\gamma\in \mathcal{K}$, such that
$$
\|\phi(t,x_0,u)\|_X \leq \beta(\|x_0\|_X,t-t_0) + \gamma(\sup_{t_0\leq\tau\leq t}\|u(\tau)\|_U)
$$
holds $\forall x_0\in X$, $\forall t\geq t_0$ and $\forall u\in U$. It is called locally input-to-state stable (LISS), if there also exists constants $\rho_{x}, \rho_{u}>0$ such that the above inequality holds $\forall x_0 :\|x_0\|_{X} \leq \rho_{x}, \forall t \geq t_0$ and $\forall u \in U :\|u\|_{U} \leq \rho_{u}$.
\end{definition}

\section{Problem formulation} \label{section:problem formulation} 
This paper studies the problem of estimating the dynamically varying probability density of large-scale stochastic agents. The dynamics of the agents are assumed to be known and satisfy the following stochastic differential equations
\begin{equation} \label{eq:Langevin equation}
    d X_i=\boldsymbol{v}(X_i, t) d t+\boldsymbol{\sigma}(X_i, t) d B_{t}, \quad i = 1,\dots,n,
\end{equation}
where $X_i\in\mathbb{R}^N$ represents the state of the $i$-th agent (e.g. positions), $\boldsymbol{v}=(v_{1}, \ldots, v_{N})\in\mathbb{R}^N$ is the deterministic dynamics (e.g. velocity fields), $B_{t}$ is an $M$-dimensional standard Wiener process, and  $\boldsymbol{\sigma}=[\sigma_{jk}]\in\mathbb{R}^{N\times M}$ represents the stochastic dynamics. We also assume the collection of states $\{X_i(t)\}_{i=1}^{n}$ are observable. 

The probability density $p(x, t)$ of agent states $\{X_i(t)\}_{i=1}^{n}$ is known to satisfy the Fokker-Planck equation \cite{pavliotis2014stochastic}:
\begin{equation} \label{eq:FP equation}
\begin{array}{cl}
     &\displaystyle\frac{\partial p(x, t)}{\partial t}=-\sum_{i=1}^{N} \frac{\partial}{\partial x_{i}}[v_{i}(x, t)p(x, t)] \\
     & \quad\quad\quad\quad\quad \displaystyle+\sum_{i=1}^{N} \sum_{j=1}^{N} \frac{\partial^{2}}{\partial x_{i} \partial x_{j}}[D_{i j}(x, t) p(x, t)],\\
     &p(\cdot,0)=p_0,
\end{array}
\end{equation}
where $p_0$ is the initial density and
$$
D_{i j}(x, t)=\frac{1}{2} \sum_{k=1}^{M} \sigma_{i k}(x, t) \sigma_{j k}(x, t).
$$
If the agent states are confined within a bounded domain $\Omega\subseteq\mathbb{R}^N$, one can impose a reflecting boundary condition
\begin{equation} \label{eq:zero-flux BC}
\mathbf{n} \cdot(\boldsymbol{g}-\boldsymbol{v}p)=0, \quad\text{on } \partial \Omega,
\end{equation}
where $\boldsymbol{g}=(\sum_{j=1}^{N}\frac{\partial}{\partial x_j}(D_{1j}p),\ldots,\sum_{j=1}^{N}\frac{\partial}{\partial x_j}(D_{Nj}p))$, $\partial\Omega$ is the boundary of $\Omega$ and $\mathbf{n}$ is the outward normal to $\partial\Omega$. 

\begin{remark} \label{remark:nonlinear and time-varying}
Note that the dynamics of agents' density \eqref{eq:FP equation} are uniquely determined by their individual dynamics \eqref{eq:Langevin equation} and are therefore known. This relationship holds even if \eqref{eq:Langevin equation} is nonlinear and time-varying. Therefore, the density filter to be presented applies to a very wide class of systems. It is worth pointing out that PDE \eqref{eq:FP equation} is always linear even if \eqref{eq:Langevin equation} is not. (It is time-varying if \eqref{eq:Langevin equation} is time-varying.)
\end{remark}

Now, we can formally state the problem to be solved as follows: 
\begin{problem}
Given the density dynamics \eqref{eq:FP equation} and agent states $\{X_i(t)\}_{i=1}^{n}$, we want to estimate their density $p(x,t)$.
\end{problem}

\section{Main Results} \label{section:main results}
In this section, we design a density filter to estimate the state (i.e. the density) of \eqref{eq:FP equation} by combining KDE and infinite-dimensional Kalman filters. Specifically, we use KDE to construct a noisy measurement of the unknown density and show that the measurement noise is approximately ``Gaussian'', which enables us to use infinite-dimensional Kalman filters to design a density filter. Stability analysis is also presented.

\subsection{Density Filter Design}
In this section, we use KDE to construct a measurement with Gaussian noise and design a density filter.

KDE is a non-parametric way to estimate an unknown probability density function (pdf) \cite{silverman1986density}. The agents' states $\{X_i(t)\}_{i=1}^{n}\subseteq\mathbb{R}^d$ at time $t$ can be seen as a set of $n$ independent samples drawn from the pdf $p(x,t)$. Fix $t$ and denote $f(x)=p(x,t)$. The density estimator is given by
\begin{equation} \label{eq:KDE}
f_n(x) = \frac{1}{n h^{d}} \sum_{i=1}^{n} K\Big(\frac{1}{h}(x-X_{i})\Big),
\end{equation}
where $K(x)$ is a kernel function \cite{silverman1986density}
and $h$ is the bandwidth, usually chosen as a function of $n$ such that $\lim _{n \rightarrow \infty} h=0$ and $\lim _{n \rightarrow \infty} n h=\infty$.
The Gaussian kernel is frequently used due to its infinite order of smoothness, given by
$$
K(x)=\frac{1}{(2 \pi)^{d/2}} \exp \Big(-\frac{1}{2}x^\intercal x\Big).
$$
It is known that the $f_n(x)$ is asymptotically normal and that $f_n(x_i)$ and $f_n(x_j)$ are asymptotically uncorrelated for any $x_i\neq x_j$. These two properties are summarized as follows.

\begin{lemma}
(\textbf{Asymptotic normality} \cite{cacoullos1966estimation}) Under conditions $\lim _{n \rightarrow \infty} h=0$ and $\lim _{n \rightarrow \infty} n h=\infty$, if the bandwidth $h$ tends to zero faster than the optimal rate, i.e.,
$$
h^{*}=o\left(\frac{1}{n}\right)^{1 /(d+4)}.
$$
then as $n\to\infty$, we have
\begin{equation} \label{eq:asymptotic normality}
    \sqrt{n h^{d}}(\hat{f}(x)-f(x)) \rightarrow \mathcal{N}\Big(0, f(x) \int[K(u)]^{2} d u\Big).
\end{equation}
\end{lemma}

\begin{lemma}
(\textbf{Asymptotic uncorrelatedness} \cite{cacoullos1966estimation}) Let $x_i$ and $x_j$ be two distinct continuity points of $f$. Then under $\lim _{n \rightarrow \infty} h=0$, as $n \rightarrow \infty$, the asymptotic covariance of $f_{n}(x_i)$ and $f_{n}(x_j)$ satisfies
$$
n h^{p} \operatorname{Cov}[f_{n}(x_i), f_{n}(x_j)] \rightarrow 0.
$$
\end{lemma}

In practice, $n$ is finite. Hence, the kernel density estimator is biased (i.e. $E[\hat{f}_n(x)]\neq f(x)$). Also, $f_{n}(x_i)$ and $f_{n}(x_j)$ are correlated. However, one can always choose a smaller bandwidth $h$ to reduce the bias and the covariance \cite{cacoullos1966estimation, silverman1986density}, which means that $f_n(x)-f(x)$ can be made approximately Gaussian with independent components when $n$ is large.
It is known that the performance of KDE largely depends on its bandwidth \cite{silverman1986density}.
Any predefined optimality of the bandwidth selection requires certain information of the unknown density $f$.
For example, the one that balances the estimation bias and variance depends on the second-order derivatives of $f$.
An excessively large (small) bandwidth may cause the problem of oversmoothing (undersmoothing).
For the density filter to be presented, we tend to choose a smaller bandwidth.
We recall that ``filters'' essentially combine past outputs to produce better estimates, which can ease the problem of undersmoothing.
In this regard, by using a density filter, we can largely circumvent the problem of optimal bandwidth selection.
Simulation studies will show that the performance of the proposed density filter is robust to different choices of bandwidth.

We are now ready to present a density filter using the infinite-dimensional Kalman filters. We rewrite the PDE of density evolution \eqref{eq:FP equation} in the form of an evolution equation and use KDE to construct a noisy measurement $y(t)$:
\begin{equation} \label{eq:evolution equation of density}
\begin{aligned}
    \Dot{p}(t) &= A(t)p(t)
    \\
    y(t)&= p_{\text{KDE}}(t) = p(t)+w(t)
\end{aligned}
\end{equation}
where $A(t)=-\sum_{i=1}^{N} \frac{\partial}{\partial x_{i}}(v_{i}\cdot)+\sum_{i=1}^{N} \sum_{j=1}^{N} \frac{\partial^{2}}{\partial x_{i} \partial x_{j}}(D_{i j}\cdot)$ is a linear operator, $p_{\text{KDE}}(t)$ represents a kernel density estimator using the states $\{X_i(t)\}_{i=1}^{n}$ at time $t$, and $w(t)$ is the measurement noise which is approximately Gaussian with covariance operator $R(t)=k\operatorname{diag}(p(t))$ where $k>0$ is a constant depending on $n$ and $h$.

According to Section \ref{section:Kalman filter}, the optimal density filter can be designed as
\begin{equation} \label{eq:optimal density filter}
    \Dot{\Hat{p}} = A(t)\Hat{p}+L(t)(y-\Hat{p}),\quad \Hat{p}(t_0)=p_{\text{KDE}}(t_0),
\end{equation}
where $L(t)=P(t) R^{-1}(t)$ is the optimal Kalman gain and $P(t)$ is a solution of the following operator Riccati equation 
\begin{equation} \label{eq:optimal Riccati}
    \dot{P}=AP+PA^{*}-PR^{-1}P.
\end{equation}

The Riccati equation \eqref{eq:optimal Riccati} actually depends on the unknown state/density $p(t)$ because $R(t)=k\operatorname{diag}(p(t))$, which means we have to approximate $R(t)$ at the same time. 
Intuitively, $\Bar{R}(t)=k\operatorname{diag}(\Hat{p}(t))$ would be a reasonable approximation. 
However, this would make \eqref{eq:optimal density filter} and \eqref{eq:optimal Riccati} strongly coupled, for which it is very difficult to analyze the stability. 
Therefore, we use $\Bar{R}(t)=\Bar{k}\operatorname{diag}(p_{\text{KDE}}(t))$, where $\Bar{k}$ is computed as $\Bar{k}=(\int[K(u)]^2du)/(nh^d)$ according to \eqref{eq:asymptotic normality}. 
In this way, we can treat the approximation error as an external disturbance and use the concept of ISS to study its stability.
We note that \eqref{eq:asymptotic normality} should be understood to be accurate in the limit ($n\to\infty$).
When $n$ is finite, the estimate may be inaccurate.
However, we shall point out that the estimation error of $\Bar{k}$ is also included as part of the approximation error of $\Bar{R}(t)$.
   
By approximating $R(t)$ with $\Bar{R}(t)=\Bar{k}\operatorname{diag}(p_{\text{KDE}}(t))$, the ``suboptimal'' density filter is correspondingly given by
\begin{equation} \label{eq:suboptimal density filter}
    \Dot{\Hat{p}} = A(t)\Hat{p}+\Bar{L}(t)(y-\Hat{p}),\quad \Hat{p}(t_0)=p_{\text{KDE}}(t_0)
\end{equation}
where $\Bar{L}(t)=\Bar{P}(t)\Bar{R}^{-1}(t)$ is the suboptimal Kalman gain and $\Bar{P}(t)$ is a solution of the approximated Riccati equation 
\begin{equation} \label{eq:suboptimal Riccati}
    \dot{\Bar{P}}=A\Bar{P}+\Bar{P}A^{*}-\Bar{P}\Bar{R}^{-1}\Bar{P}.
\end{equation}

The optimal density filter \eqref{eq:optimal density filter} essentially combines linearly all past density estimations $\{\Hat{p}(\tau)\}_{t_0\leq\tau<t}$ and the most recent KDE measurement $p_{\text{KDE}}(t)$ to produce a density estimate $\Hat{p}(t)$ with minimum estimation error covariance. The approximated density filter \eqref{eq:suboptimal density filter} is called ``suboptimal'' in the sense that the suboptimal gain $\Bar{L}$ remains ``close'' to the optimal gain $L$ (which will be proved in Corollary \ref{corollary:gain remains close}).

\subsection{Stability Analysis of the Suboptimal Density Filter}
In this section, we study the stability of the suboptimal density filter \eqref{eq:suboptimal density filter} and the associated Riccati equation \eqref{eq:suboptimal Riccati}. 

Let $P_*>0$ be the minimum (but unknown) covariance at $t=t_0$. We denote by $\Pi(t)$ the solution of \eqref{eq:optimal Riccati} with the initial condition $P_*$, which represents the flow of minimum covariance and also generates the optimal gain. Let $P_0>0$ be a ``guessed'' initial condition. 
Denote by $\Bar{\Pi}(t)$ the solution of \eqref{eq:suboptimal Riccati} with the initial condition $P_0$. We will show that $\Bar{\Pi}(t)$ converges to and remains close to $\Pi(t)$ in the presence of approximation error on $R$.

Define $\Gamma=\Bar{\Pi}-\Pi$. Using \eqref{eq:optimal Riccati} and \eqref{eq:suboptimal Riccati} we have
\begin{equation} \label{eq:Riccati error equation1}
\begin{aligned}
    &\Dot{\Gamma} = A\Gamma+\Gamma A^{*} - \Bar{\Pi}\Bar{R}^{-1}\Bar{\Pi} + \Pi R^{-1}\Pi,\\
    &\Gamma(t_0) =P_0-P_*.
\end{aligned}
\end{equation}
Define $\Tilde{p} = \Hat{p}-p$. Then along $\Bar{\Pi}(t)$ we have
\begin{equation} \label{eq:estimation error equation}
    \Dot{\Tilde{p}} = (A-\Bar{\Pi}\Bar{R}^{-1})\Tilde{p}+\Bar{\Pi}\Bar{R}^{-1}w.
\end{equation}

Our idea is to show that \eqref{eq:Riccati error equation1} is ISS with respect to the approximation error on $R$ (more precisely $\|\Bar{R}^{-1}-R^{-1}\|$, which equals $0$ if and only if $R=\Bar{R}$). In this way, the suboptimal gain $\Bar{L}$ remains close to $L$ when there is any approximation error, and converges to $0$ when the approximation error vanishes. We also show that the suboptimal density filter \eqref{eq:suboptimal density filter} is stable even though we use an approximation for $R$.

\begin{remark} \label{remark:true covariance equation}
To avoid confusions, we point out that if we use \eqref{eq:optimal Riccati} to compute the gain for \eqref{eq:optimal density filter}, then the estimation error covariance $\operatorname{Cov}[\Tilde{p},\Tilde{p}]$ along $\Pi(t)$ also satisfies \eqref{eq:optimal Riccati}, which is a property of Kalman filters \cite{kalman1961new}. 
However, if we use \eqref{eq:suboptimal Riccati} to compute the gain for \eqref{eq:suboptimal density filter}, then the covariance $\operatorname{Cov}[\Tilde{p},\Tilde{p}]$ along $\Bar{\Pi}(t)$, denoted by $Q(t)$ in this case, actually satisfies
\begin{equation} \label{eq:true covariance equation}
    \dot{Q}=AQ+QA^{*}-\Bar{\Pi}(t)\Bar{R}^{-1}R\Bar{R}^{-1}\Bar{\Pi}(t).
\end{equation}
It would be desirable to prove that the solutions of \eqref{eq:optimal Riccati} and \eqref{eq:true covariance equation} remain close, which is much harder because it involves three Riccati equations and will be left as our future work.
\end{remark}

The following assumption is required for proving stability.

\begin{assumption} \label{assumption:uniform boundedness}
Assume that $\|\Pi(t)\|$ and $\|\Bar{\Pi}(t)\|$ are uniformly bounded, and that there exist positive constants $c_1$ and $c_2$ such that for all $t\geq t_0$,
\begin{equation}\label{eq:uniform positivity}
    0<c_1I\leq R^{-1}(t),\Bar{R}^{-1}(t),\Pi^{-1}(t),\Bar{\Pi}^{-1}(t)\leq c_2I.
\end{equation}
\end{assumption}

\begin{remark} \label{remark:uniform boundedness}
The assumption for $R^{-1}(t)$ is an imposed requirement for $p(t)$ considering $R(t)=k\operatorname{diag}(p(t))$, which roughly speaking, requires that $p(t)$ has positive upper bounds and lower bounds. 
The assumption for $\Bar{R}^{-1}(t)$ can be easily satisfied because $\Bar{R}(t)=\Bar{k}\operatorname{diag}(p_{\text{KDE}}(t))$ and we construct $p_{\text{KDE}}(t)$. 
For finite-dimensional systems, the assumptions for $\|\Pi(t)\|$, $\|\Bar{\Pi}(t)\|$, $\Pi^{-1}(t)$ and $\Bar{\Pi}^{-1}(t)$ follow the assumption for $R^{-1}(t)$ and $\Bar{R}^{-1}(t)$. 
This is because when $\Pi(t)$ and $\Bar{\Pi}(t)$ are large, the solutions of \eqref{eq:optimal Riccati} and \eqref{eq:suboptimal Riccati} will be dominated by the negative second-order term and decay. 
Although intuitively correct, we find it much harder to prove for the infinite-dimensional case and thus leave it as our future work.
\end{remark}

Stability results for \eqref{eq:Riccati error equation1} and \eqref{eq:estimation error equation} are given as follows.

\begin{theorem} \label{thm:error Riccati LISS}
Under Assumption \ref{assumption:uniform boundedness}, the unforced part of \eqref{eq:estimation error equation}, given by
\begin{equation} \label{eq:unforced estimation error equation1}
    \Dot{\Tilde{p}} = (A-\Bar{\Pi}\Bar{R}^{-1})\Tilde{p},
\end{equation}
is uniformly exponentially stable, and \eqref{eq:Riccati error equation1} is locally input-to-state stable (LISS) with respect to the approximation error in the form of $\|\Bar{R}^{-1}-R^{-1}\|$.
\end{theorem}

\begin{proof}
To prove the first statement, consider a Lyapunov functional $V_1 = \langle \Bar{\Pi}^{-1}\Tilde{p},\Tilde{p} \rangle$.
We have
\begin{align*}
    \Dot{V}_1 &= \big\langle\Bar{\Pi}^{-1}\Dot{\Tilde{p}},\Tilde{p}\big\rangle + \big\langle\Bar{\Pi}^{-1}\Tilde{p},\Dot{\Tilde{p}}\big\rangle - \big\langle\Bar{\Pi}^{-1}\Dot{\Bar{\Pi}}\Bar{\Pi}^{-1}\Tilde{p},\Tilde{p}\big\rangle\\
    &= \big\langle\Bar{\Pi}^{-1}(A-\Bar{\Pi}\Bar{R}^{-1})\Tilde{p},\Tilde{p}\big\rangle + \big\langle(A^*-\Bar{R}^{-1}\Bar{\Pi})\Bar{\Pi}^{-1}\Tilde{p},\Tilde{p}\big\rangle\\
    & \quad -\big\langle \Bar{\Pi}^{-1}(A\Bar{\Pi}+\Bar{\Pi}A^{*}-\Bar{\Pi}\Bar{R}^{-1}\Bar{\Pi})\Bar{\Pi}^{-1}\Tilde{p},\Tilde{p} \big\rangle\\
    &= -\langle \Bar{R}^{-1}\Tilde{p},\Tilde{p} \rangle.
\end{align*}
In view of \eqref{eq:uniform positivity}, we conclude that \eqref{eq:unforced estimation error equation1} is uniformly exponentially stable. 
Similarly, by considering a Lyapunov functional defined by $V_2=\langle\Pi^{-1}\Tilde{p},\Tilde{p}\rangle$, one can show that the following system along $\Pi(t)$ is also uniformly exponentially stable:
\begin{equation} \label{eq:unforced estimation error equation2}
    \Dot{\Tilde{p}} = (A-\Pi R^{-1})\Tilde{p}.
\end{equation}
We also note that the inner products in $V_1$ and $V_2$ are equivalent because of the assumption \eqref{eq:uniform positivity}.
To prove the second statement, we rewrite \eqref{eq:Riccati error equation1} as
\begin{equation}\label{eq:Riccati error equation2}
\begin{aligned}
    \Dot{\Gamma}
    &=A\Gamma+\Gamma A^{*} - \Bar{\Pi}\Bar{R}^{-1}\Gamma - \Bar{\Pi} \Bar{R}^{-1}\Pi - \Gamma R^{-1}\Pi + \Bar{\Pi}R^{-1}\Pi\\
    &=(A-\Bar{\Pi}\Bar{R}^{-1})\Gamma + \Gamma(A^{*}-R^{-1}\Pi) - \Bar{\Pi}(\Bar{R}^{-1}-R^{-1})\Pi
\end{aligned}
\end{equation}
We note that \eqref{eq:Riccati error equation2} is essentially a linear equation. 
Now fix $q$ with $\|q\|=1$. 
Since \eqref{eq:unforced estimation error equation1} and \eqref{eq:unforced estimation error equation2} are uniformly exponentially stable, and $\|\Bar{\Pi}\|$ and $\|\Pi\|$ are assumed to be uniformly bounded, there exist constants $\lambda,c>0$ such that
\begin{align*}
    \quad\|\Gamma(t)q\|
    &\leq e^{-\lambda(t-t_0)}\|\Gamma(t_0)q\| \\
    &\quad+\int_{t_0}^{t}e^{-\lambda(t-\tau)}c\|\bar{R}^{-1}(\tau)-R^{-1}(\tau)\|\|q\|d\tau \\
\end{align*}
\begin{align*}
    &\leq e^{-\lambda(t-t_0)}\|\Gamma(t_0)q\| \\
    &\quad+c\sup_{t_0\leq\tau\leq t}\|\bar{R}^{-1}(\tau)-R^{-1}(\tau)\|\int_{t_0}^{t}e^{-\lambda(t-s)}ds \\
    &\leq e^{-\lambda(t-t_0)}\|\Gamma(t_0)q\|+\frac{c}{\lambda}\sup_{t_0\leq\tau\leq t}\|\bar{R}^{-1}(\tau)-R^{-1}(\tau)\| \\
    &=:\beta(\|\Gamma(t_0)q\|,t-t_0) + \gamma(\sup_{t_0\leq\tau\leq t}\|\Bar{R}^{-1}(\tau)-R^{-1}(\tau)\|),
\end{align*}
where $\beta\in\mathcal{KL}$ and $\gamma\in\mathcal{K}$. Under the uniform boundedness assumption of $\Pi$ and $\Bar{\Pi}$, we conclude that \eqref{eq:Riccati error equation1} is LISS.
\end{proof}

We can conclude from Theorem \ref{thm:error Riccati LISS} that the suboptimal gain $\Bar{L}$ also remains close to the optimal gain $L$, given as follows.

\begin{corollary}\label{corollary:gain remains close}
Under Assumption \ref{assumption:uniform boundedness}, there exist functions $\beta_1\in\mathcal{KL}$ and $\gamma_1\in\mathcal{K}$ such that
$$
\|\Bar{L}-L\|\leq\beta_1(\|\Gamma(t_0)\|,t-t_0) + \gamma_1(\|\Bar{R}^{-1}-R^{-1}\|).
$$
\end{corollary}

\begin{proof}
Observe that
\begin{align*}
    \|\Bar{L}-L\| &=\|\Bar{\Pi}\Bar{R}^{-1}-\Pi R^{-1}\| \\
    &\leq \|\Bar{\Pi}\Bar{R}^{-1}-\Pi\Bar{R}^{-1} + \Pi\Bar{R}^{-1}-\Pi R^{-1}\| \\
    &\leq \|\Bar{R}^{-1}\|\|\Bar{\Pi}-\Pi\| + \|\Pi\|\|\Bar{R}^{-1}-R^{-1}\|
\end{align*}
In view of Definition \ref{definition:(L)ISS}, Theorem \ref{thm:error Riccati LISS} and the uniform boundedness of $\|\Pi\|$ and $\|\Bar{R}^{-1}\|$, we obtain the desired result.
\end{proof}

\section{Simulation Studies} \label{section:simulation}
In this section, we study the performance of the proposed density filter. Consider a collection of $300$ agents given by
\begin{equation}
    dX_i=\nabla\cdot\frac{D\nabla f(x)}{f(x)} dt+DdB_t, \quad i = 1,\dots,300,
\end{equation}
where the states $\{X_i\}_{i=1}^{300}$ are restricted within $\Omega=[0,1]^2$, $D=0.05$ and $f(x)$ is a continuous pdf over $\Omega$ to be specified. The initial conditions $\{X_i(t_0)\}_{i=1}^{300}$ are drawn from the uniform distribution over $\Omega$. Therefore, the ground truth density of the agents satisfies
\begin{align*} 
\begin{split}
     &\partial_t p(x,t) =-\nabla\cdot\frac{Dp(x,t)\nabla f(x)}{f(x)} + \frac{1}{2}D^2\Delta p(x,t), \\
     &p(\cdot,t_0)=1,
\end{split}
\end{align*}
with a reflecting boundary condition \eqref{eq:zero-flux BC}. For illustration purpose, we let $f(x)$ be a Gaussian mixture of two components with the common covariance matrix $\text{diag}(0.015,0.015)$ but different time-varying means $[0.5+0.35\cos(0.2t),0.5+0.35\sin(0.2t)]^\intercal$ and $[0.5+0.35\cos(0.2t+\pi),0.5+0.35\sin(0.2t+\pi)]^\intercal$. Under this design, the agents are nonlinear and time-varying and their states will converge to the two ``spinning'' Gaussian components.
 
We use the finite difference method \cite{chang1970practical} to numerically solve the density filter \eqref{eq:suboptimal density filter} and the operator Riccati equation \eqref{eq:suboptimal Riccati}. Specifically, $\Omega$ is equally divided into a $30\times30$ grid. The pdfs $\Hat{p}$ and $p_{\text{KDE}}$ are represented as $900\times1$ vectors. The operators $A$, $\Bar{P}$ and $\Bar{R}$ are represented as $900\times900$ matrices. We set the initial conditions to be $\Bar{P}(t_0) = I$ and $\Hat{p}(t_0) = p_{\text{KDE}}(t_0)$. The time difference is $dt=0.1s$ and the bandwidth is $h=0.05$. Note that $A$ is highly sparse and $\Bar{R}$ is diagonal, so the computation is very fast in general. 

Simulation results are given in Fig. \ref{fig:density filter}. We see that the estimate by KDE is usually rugged, especially in the early stage when the true density is nearly uniform. (If we had known that it is nearly uniform, we would have chosen a much larger bandwidth $h$ for better smoothing.) In the late stage, the estimate by KDE has two major components since the samples are more concentrated. But it still has undesired modes due to outliers. The density filter, by taking advantage of the dynamics, quickly "recognizes" the two underlying Gaussian components and gradually catches up with the evolution of the ground truth density. In Fig. \ref{fig:error}, we compare the $L^2$ norms of estimation errors of the KDE and the density filter under different choices of bandwidth. It shows that the estimation error of the density filter quickly converges and the performance is robust to different choices of bandwidth.

\begin{figure}[hbt!]
    \centering
    \includegraphics[width=0.9\columnwidth]{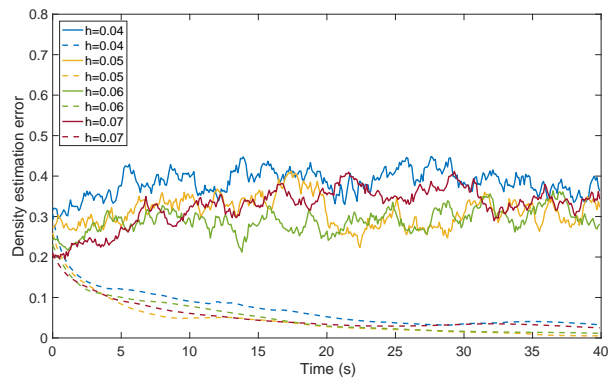}
    \caption{Estimation errors of the KDE (solid line) and the density filter (dashed line) with different choices of bandwidth.}
    \label{fig:error}
\end{figure}

\begin{figure*}[t]
\setlength{\belowcaptionskip}{-0.5cm}
    \centering
    \begin{subfigure}[b]{0.22\textwidth}
        \centering
        \includegraphics[width=0.95\textwidth]{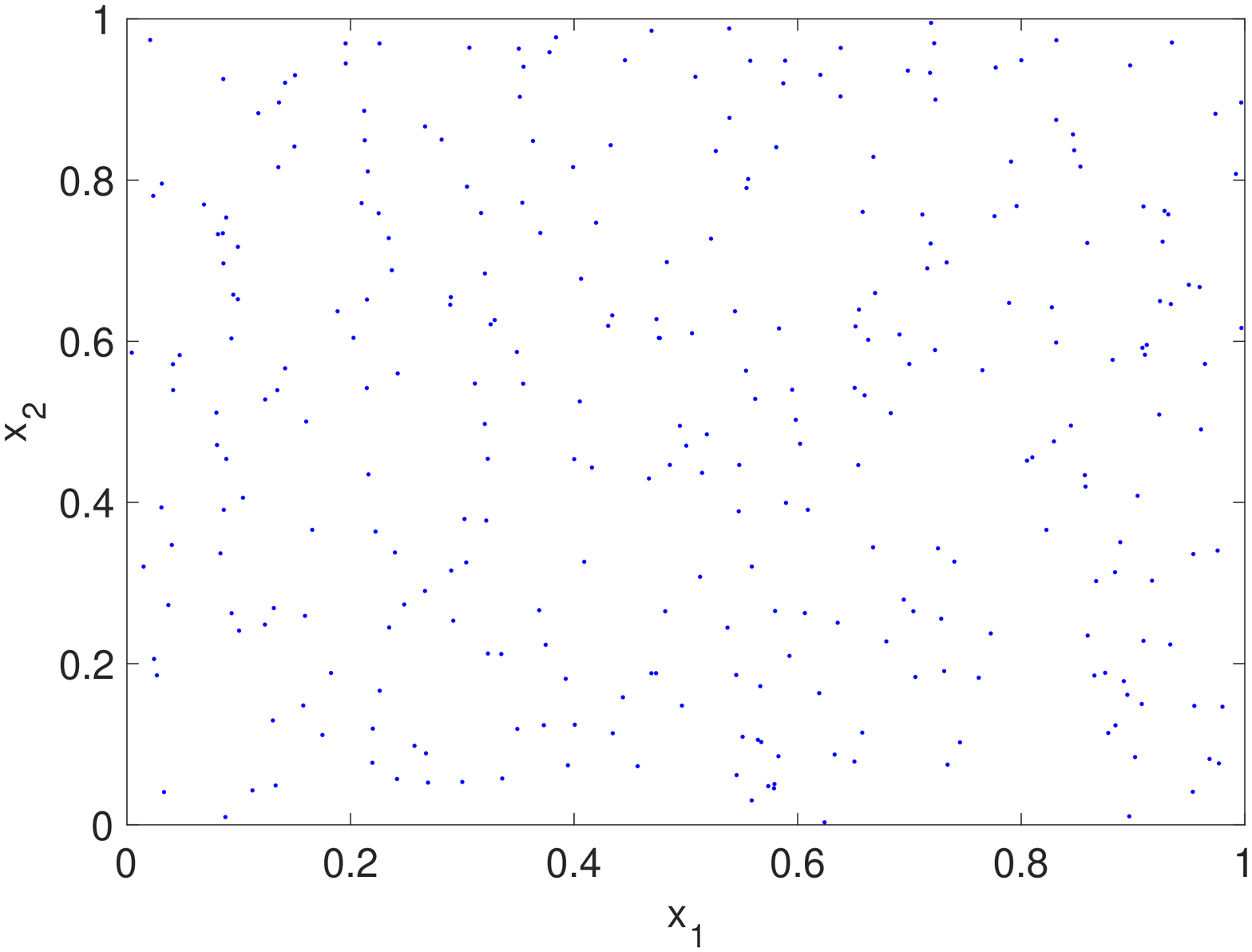}
    \end{subfigure}
    \begin{subfigure}[b]{0.22\textwidth}
        \centering
        \includegraphics[width=0.95\textwidth]{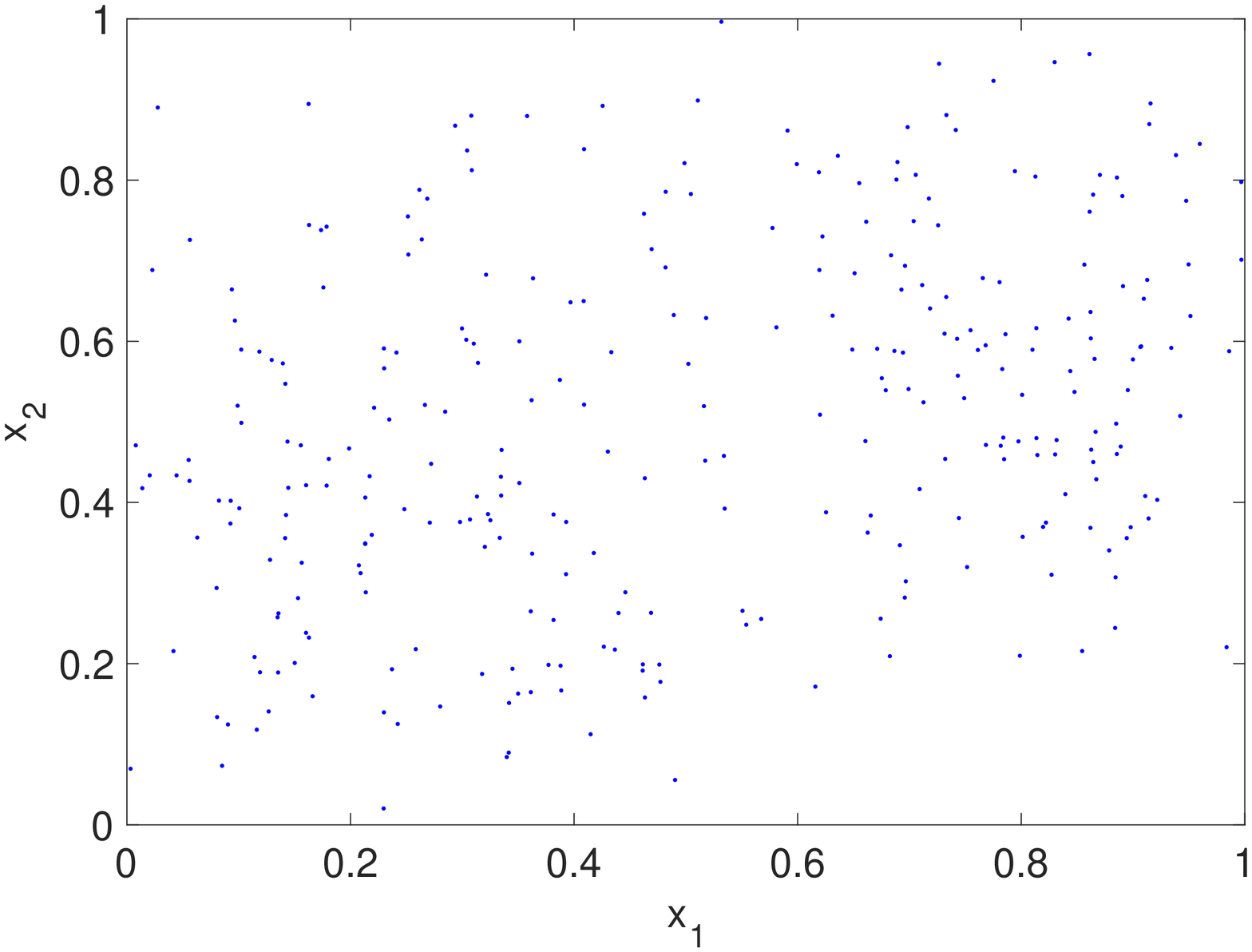}
    \end{subfigure}
    \begin{subfigure}[b]{0.22\textwidth}
        \centering
        \includegraphics[width=0.95\textwidth]{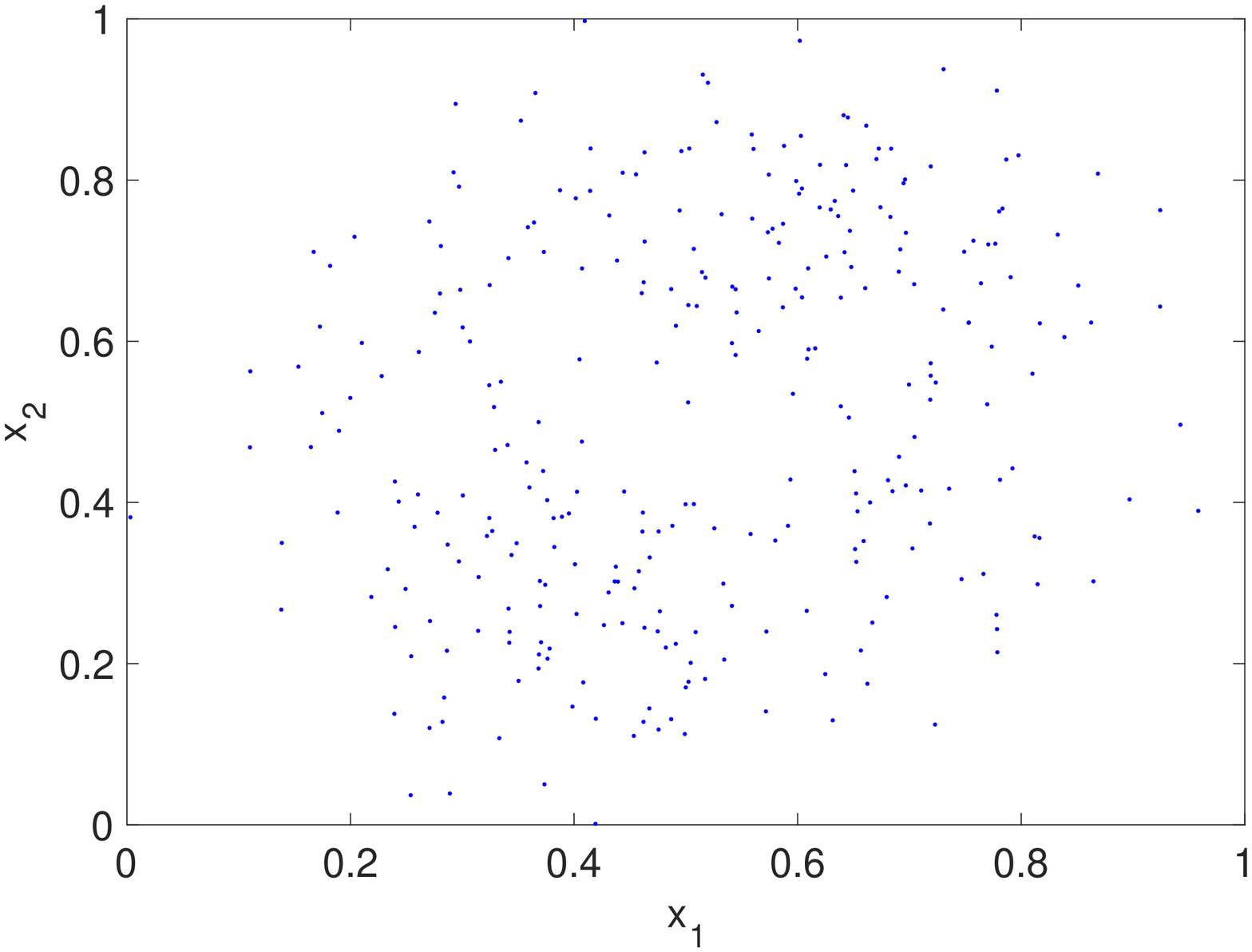}
    \end{subfigure}
    \begin{subfigure}[b]{0.22\textwidth}
        \centering
        \includegraphics[width=0.95\textwidth]{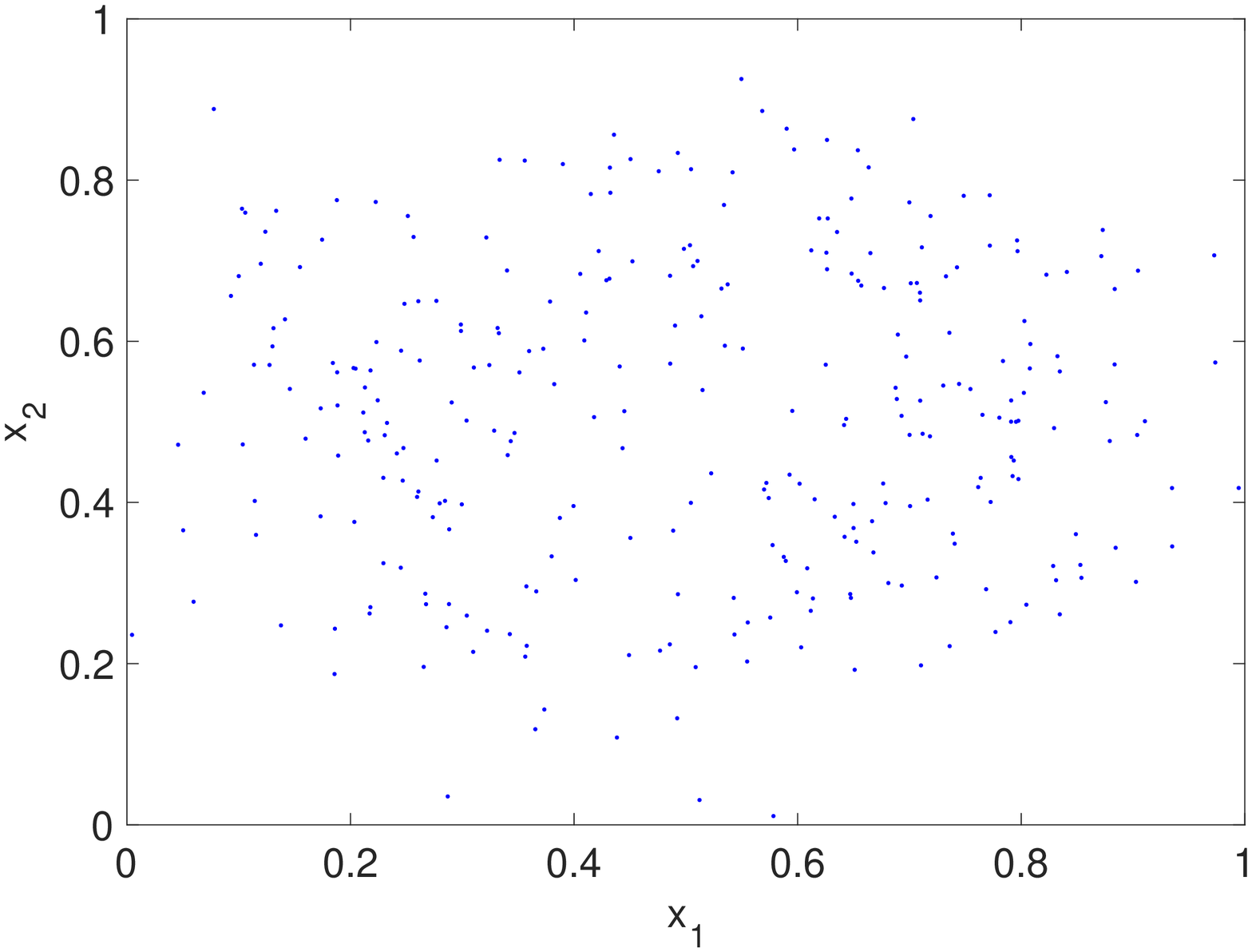}
    \end{subfigure}
    \vspace{3mm}

    \begin{subfigure}[b]{0.22\textwidth}
        \centering
        \includegraphics[width=0.95\textwidth]{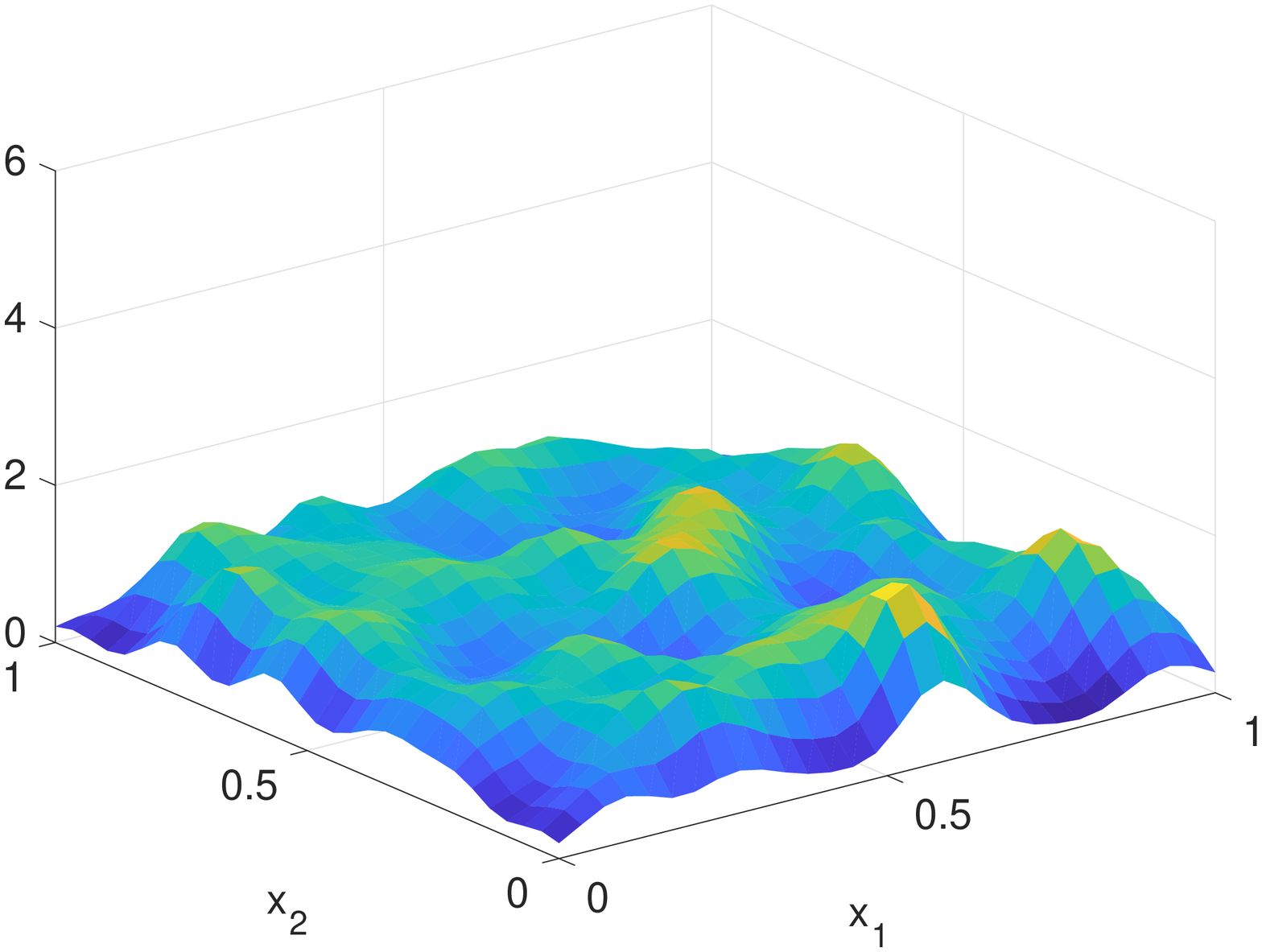}
    \end{subfigure}
    \begin{subfigure}[b]{0.22\textwidth}
        \centering
        \includegraphics[width=0.95\textwidth]{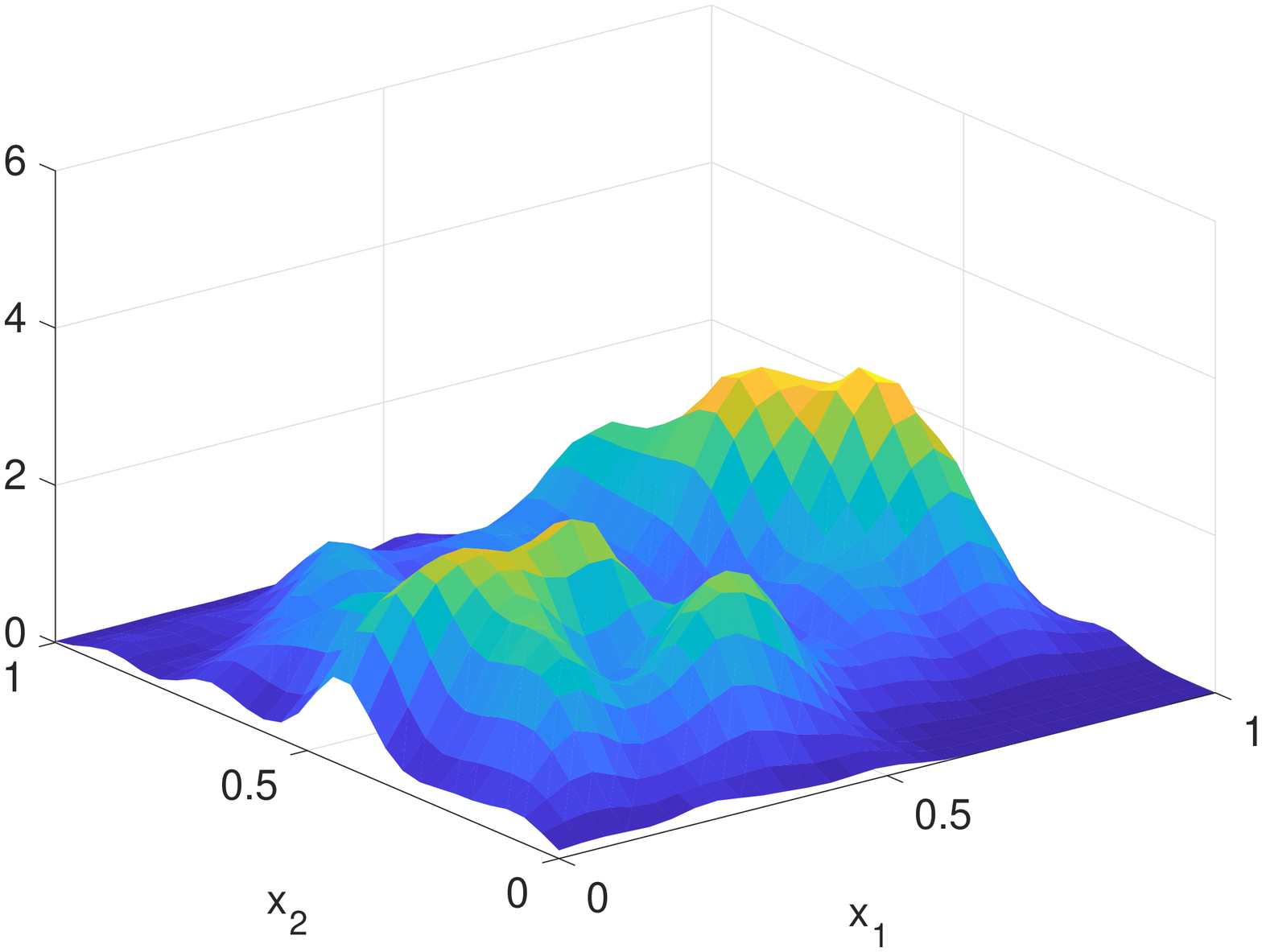}
    \end{subfigure}
    \begin{subfigure}[b]{0.22\textwidth}
        \centering
        \includegraphics[width=0.95\textwidth]{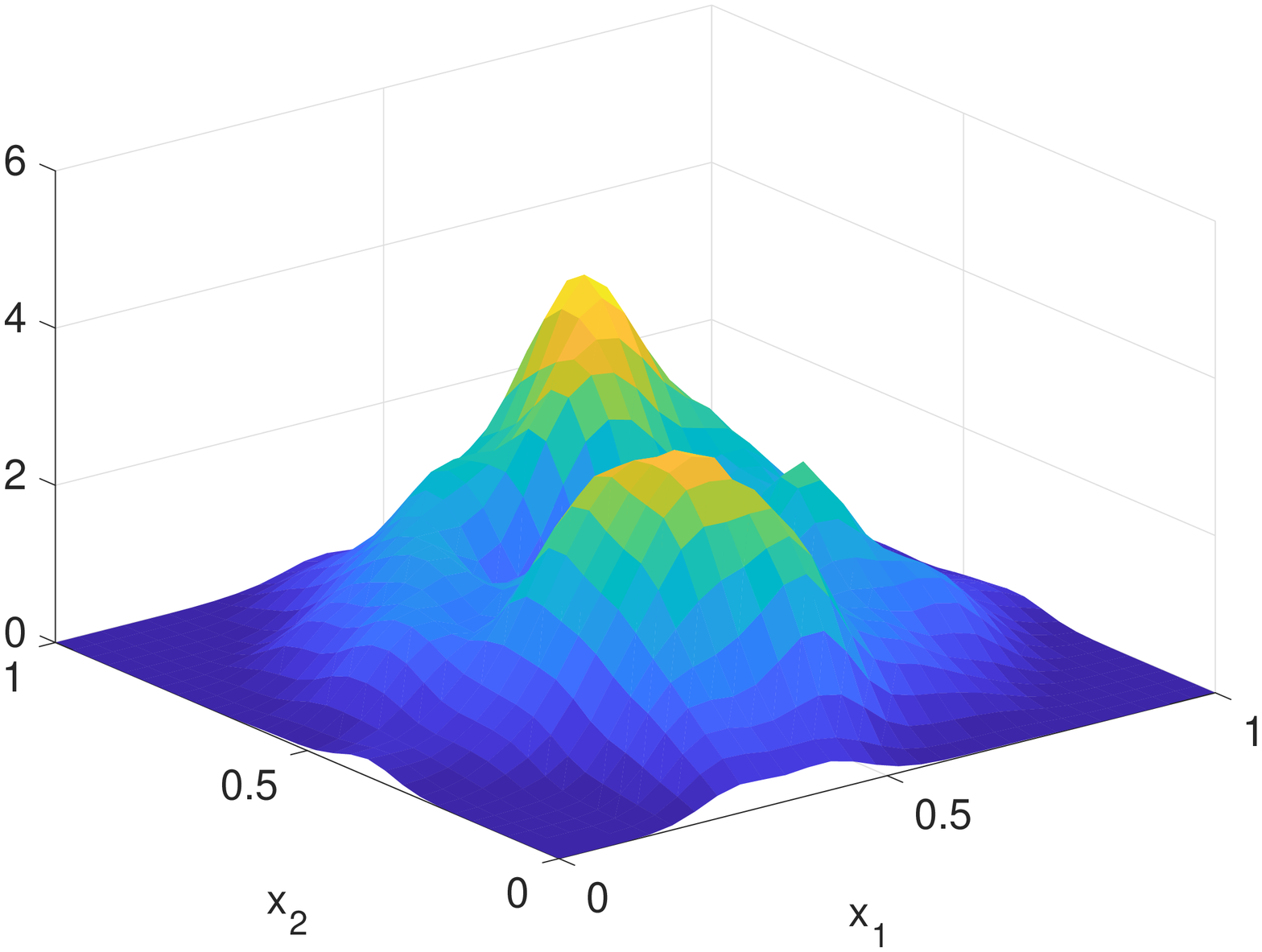}
    \end{subfigure}
    \begin{subfigure}[b]{0.22\textwidth}
        \centering
        \includegraphics[width=0.95\textwidth]{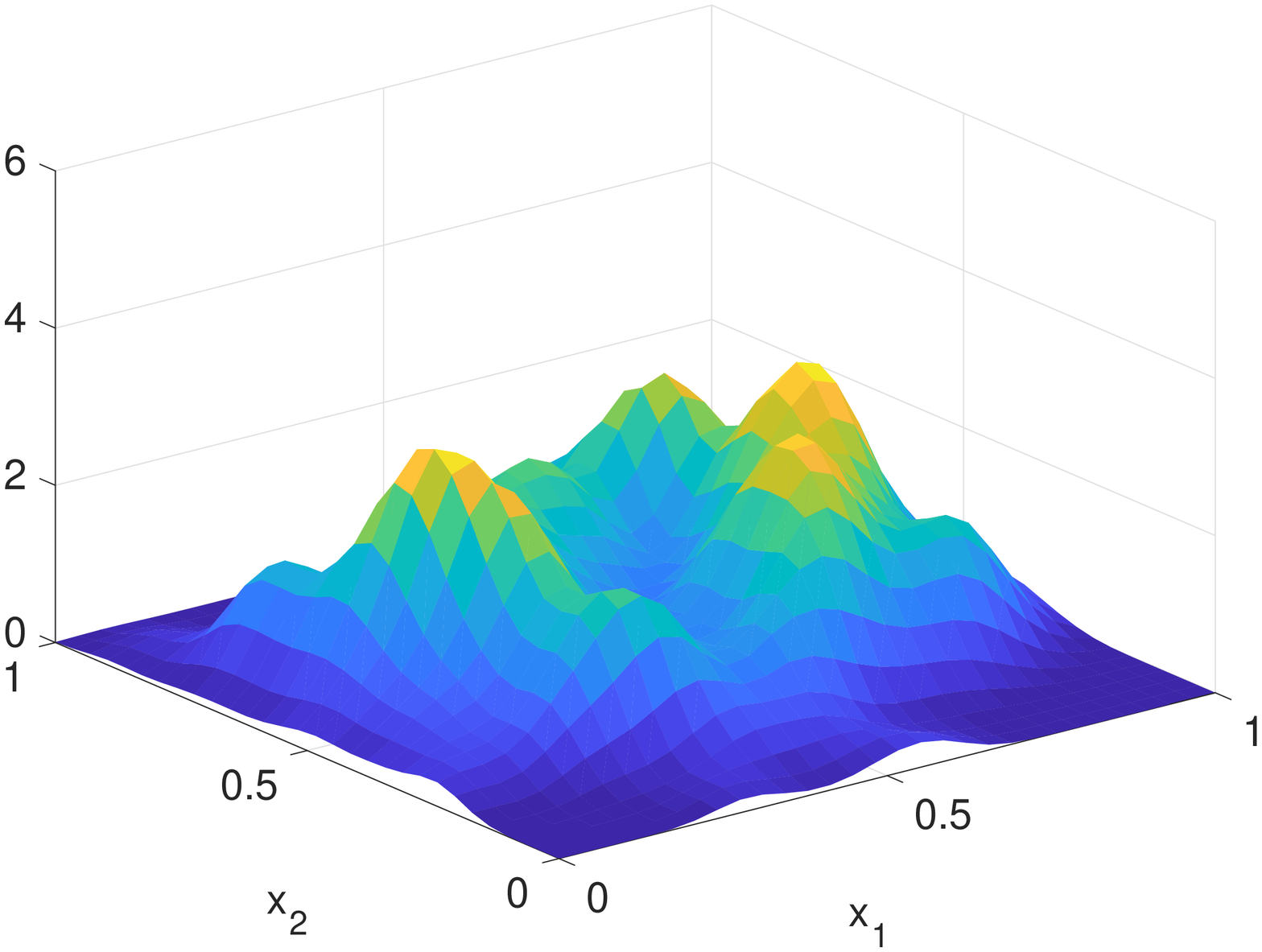}
    \end{subfigure}
    \vspace{3mm}

    \begin{subfigure}[b]{0.22\textwidth}
        \centering
        \includegraphics[width=0.95\textwidth]{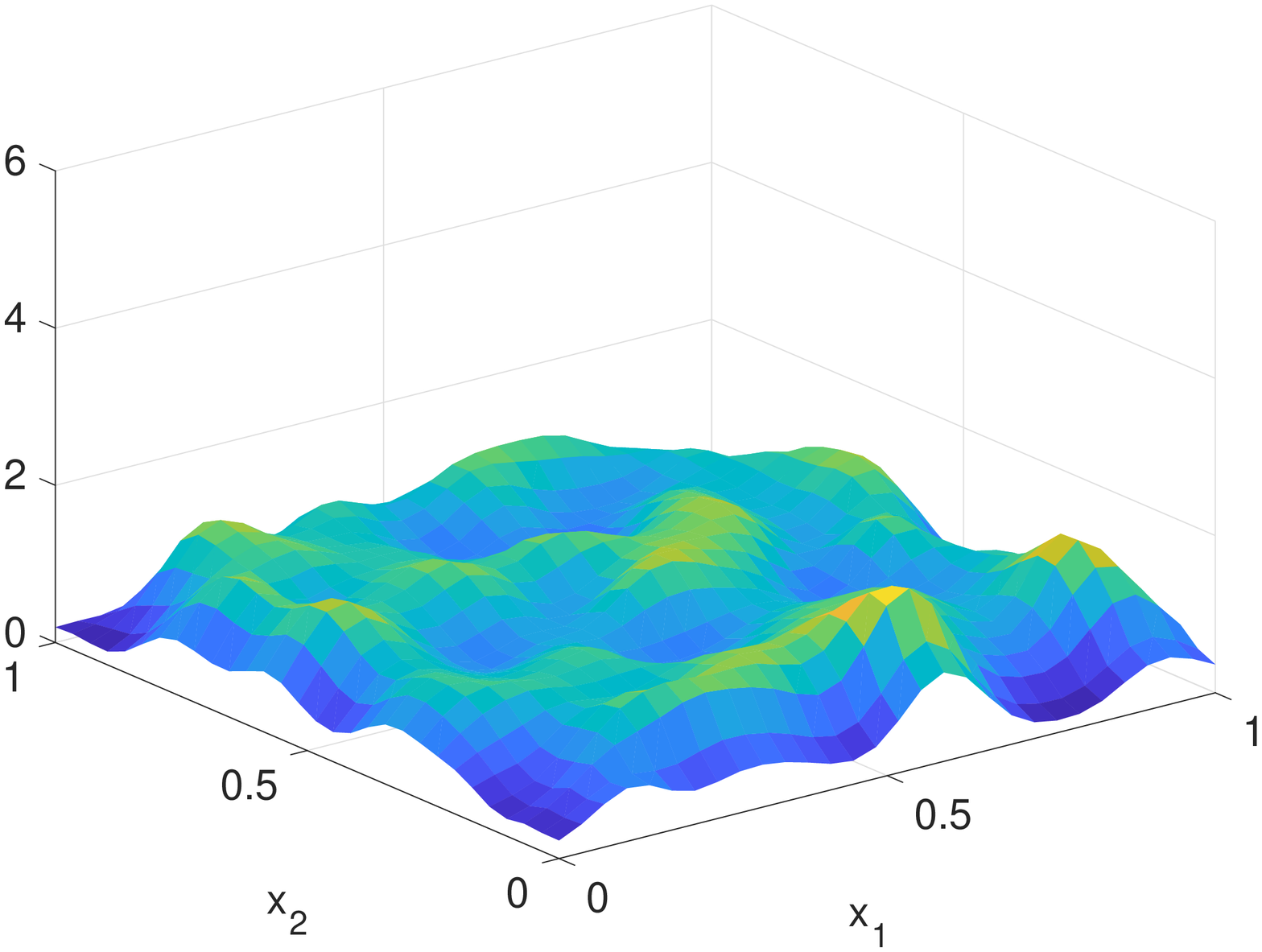}
    \end{subfigure}
    \begin{subfigure}[b]{0.22\textwidth}
        \centering
        \includegraphics[width=0.95\textwidth]{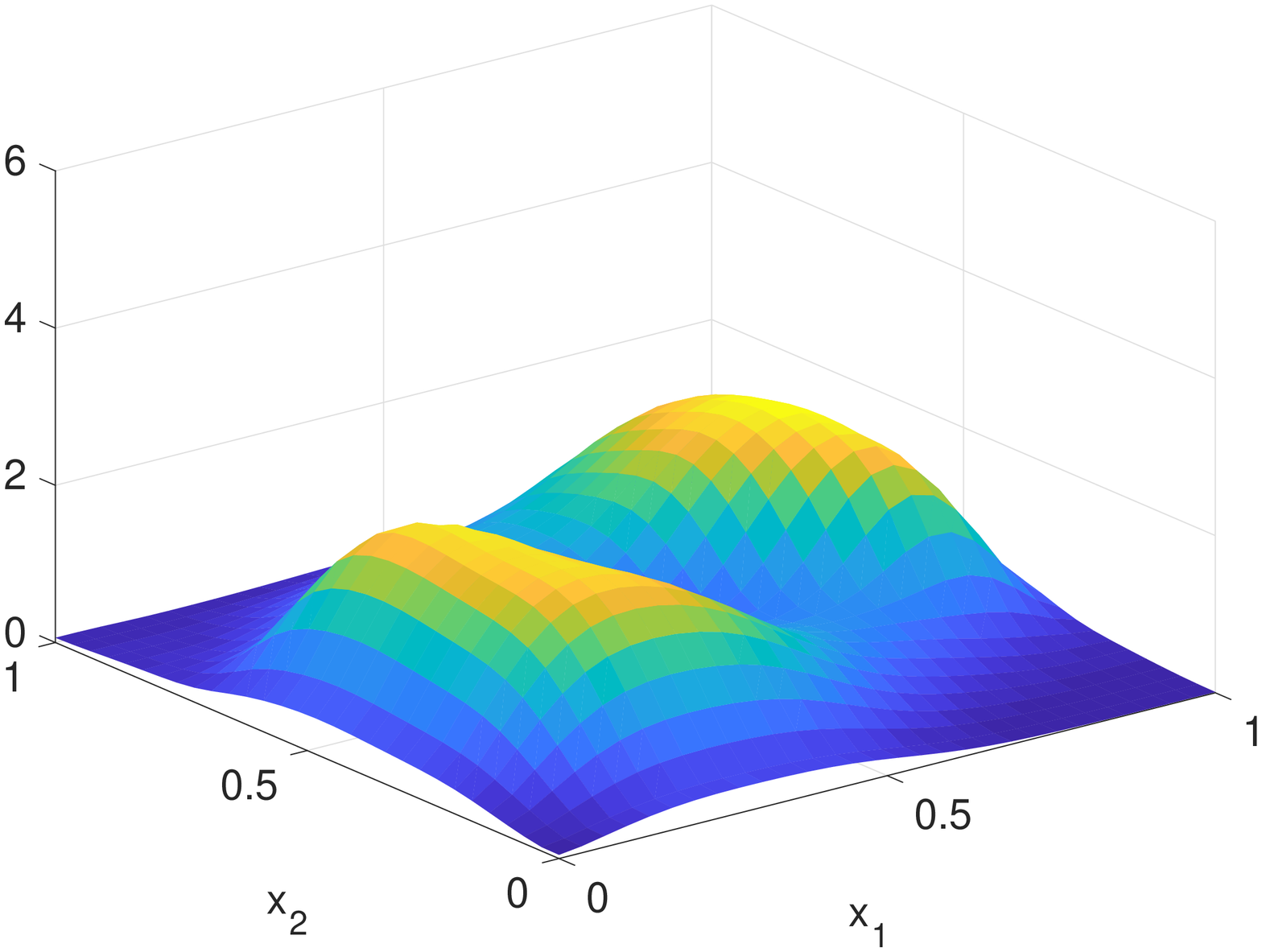}
    \end{subfigure}
    \begin{subfigure}[b]{0.22\textwidth}
        \centering
        \includegraphics[width=0.95\textwidth]{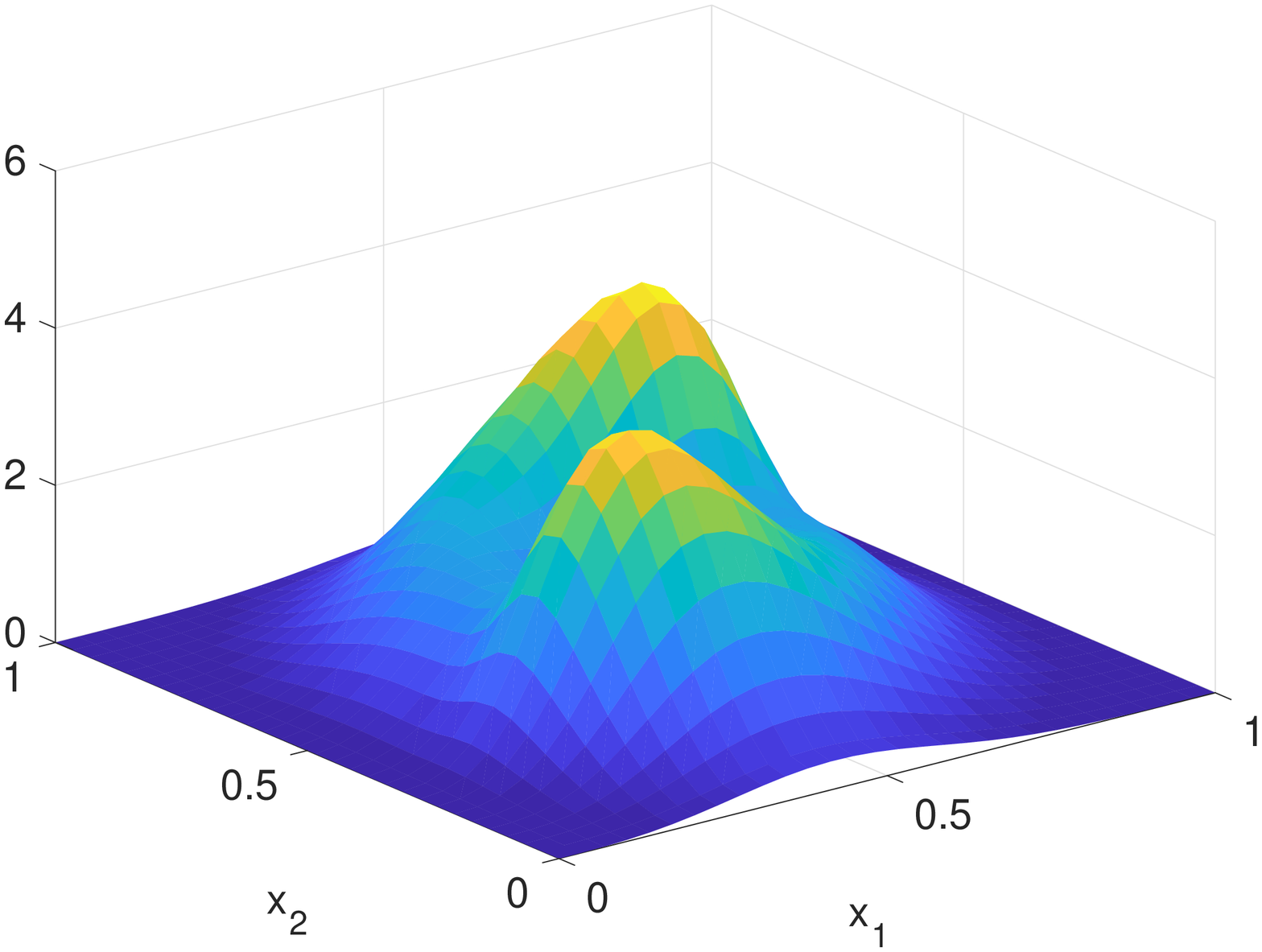}
    \end{subfigure}
    \begin{subfigure}[b]{0.22\textwidth}
        \centering
        \includegraphics[width=0.95\textwidth]{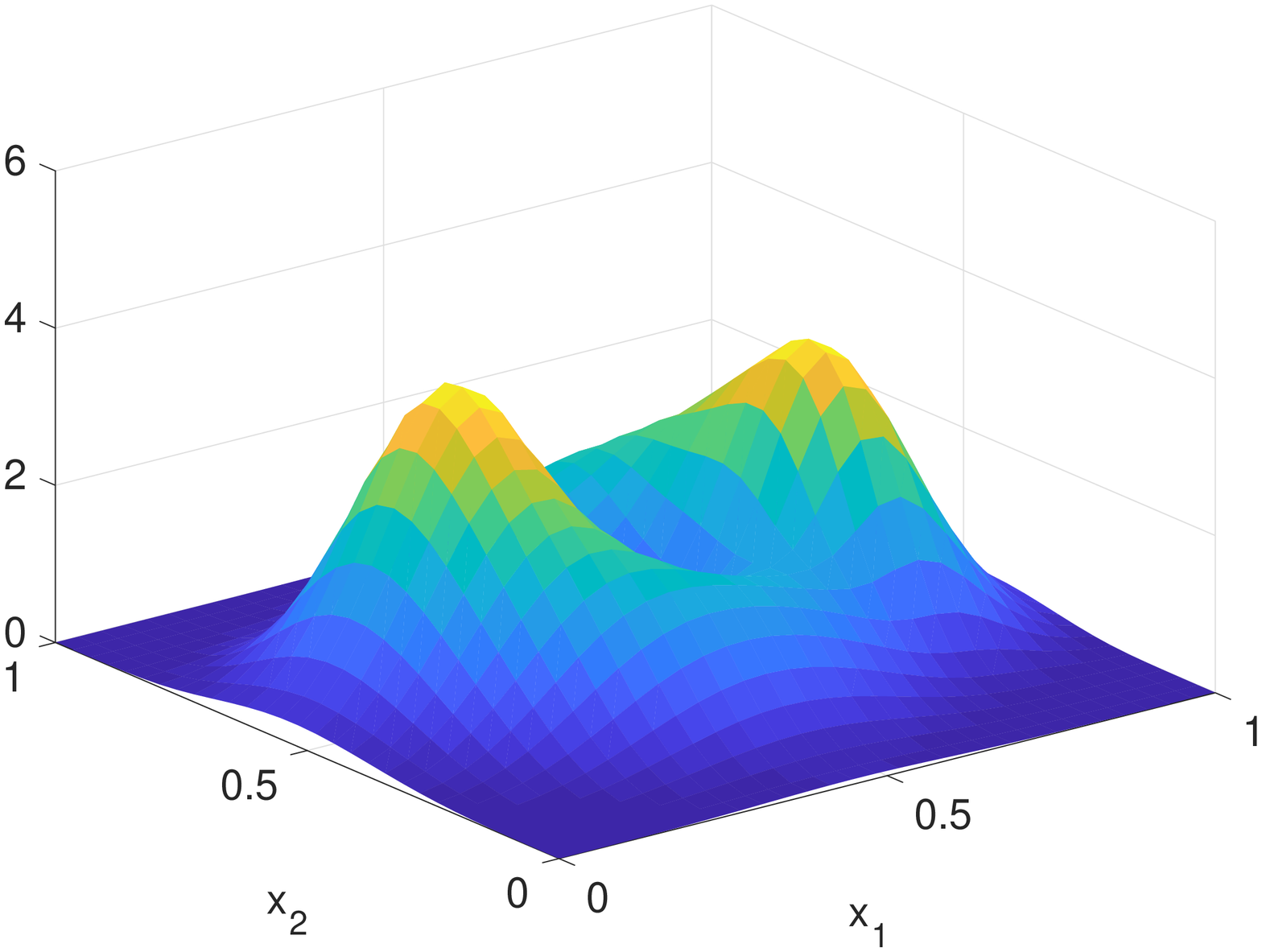}
    \end{subfigure}
    \vspace{3mm}
    
    \begin{subfigure}[b]{0.22\textwidth}
        \centering
        \includegraphics[width=0.95\textwidth]{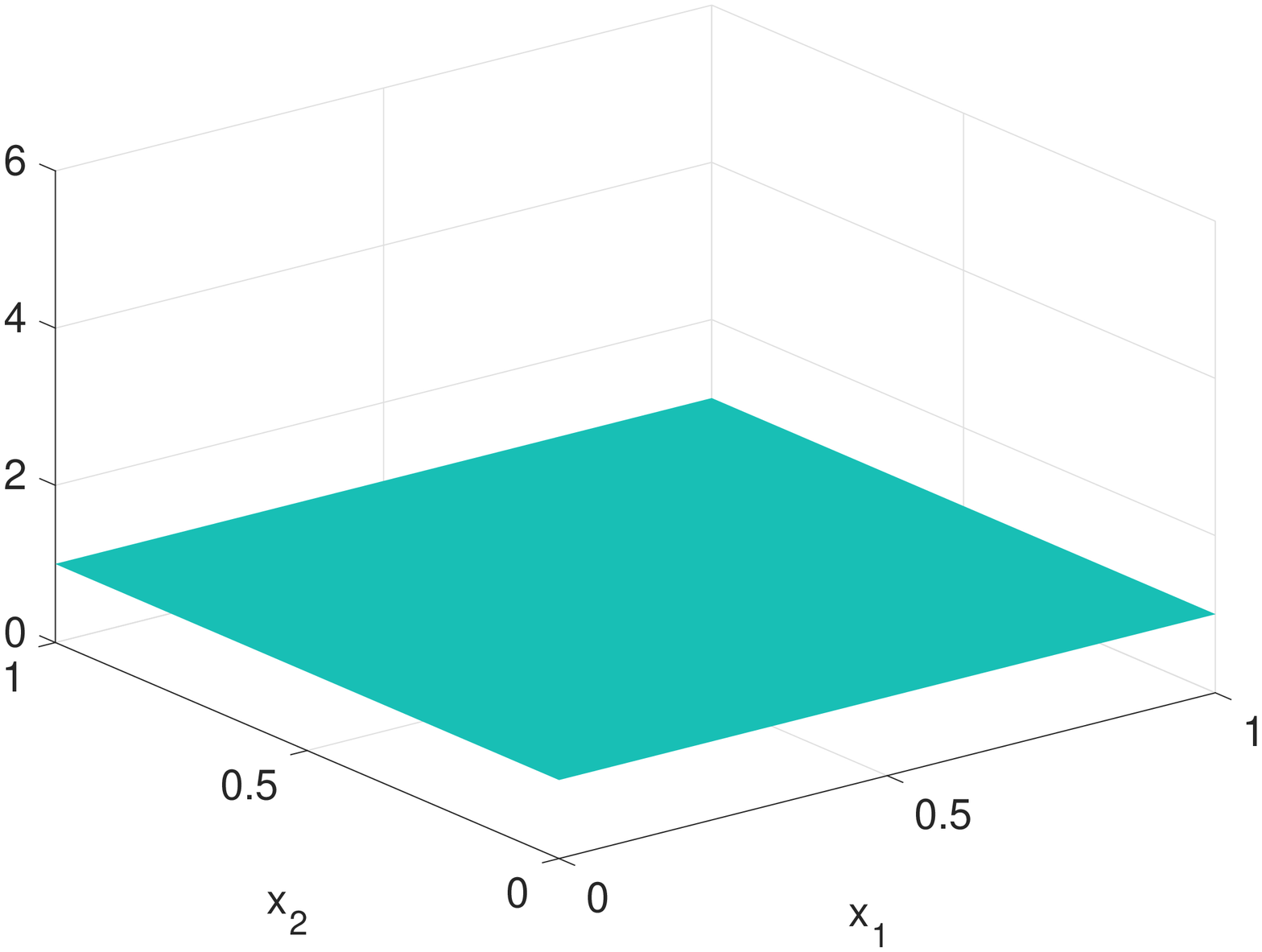}
    \end{subfigure}
    \begin{subfigure}[b]{0.22\textwidth}
        \centering
        \includegraphics[width=0.95\textwidth]{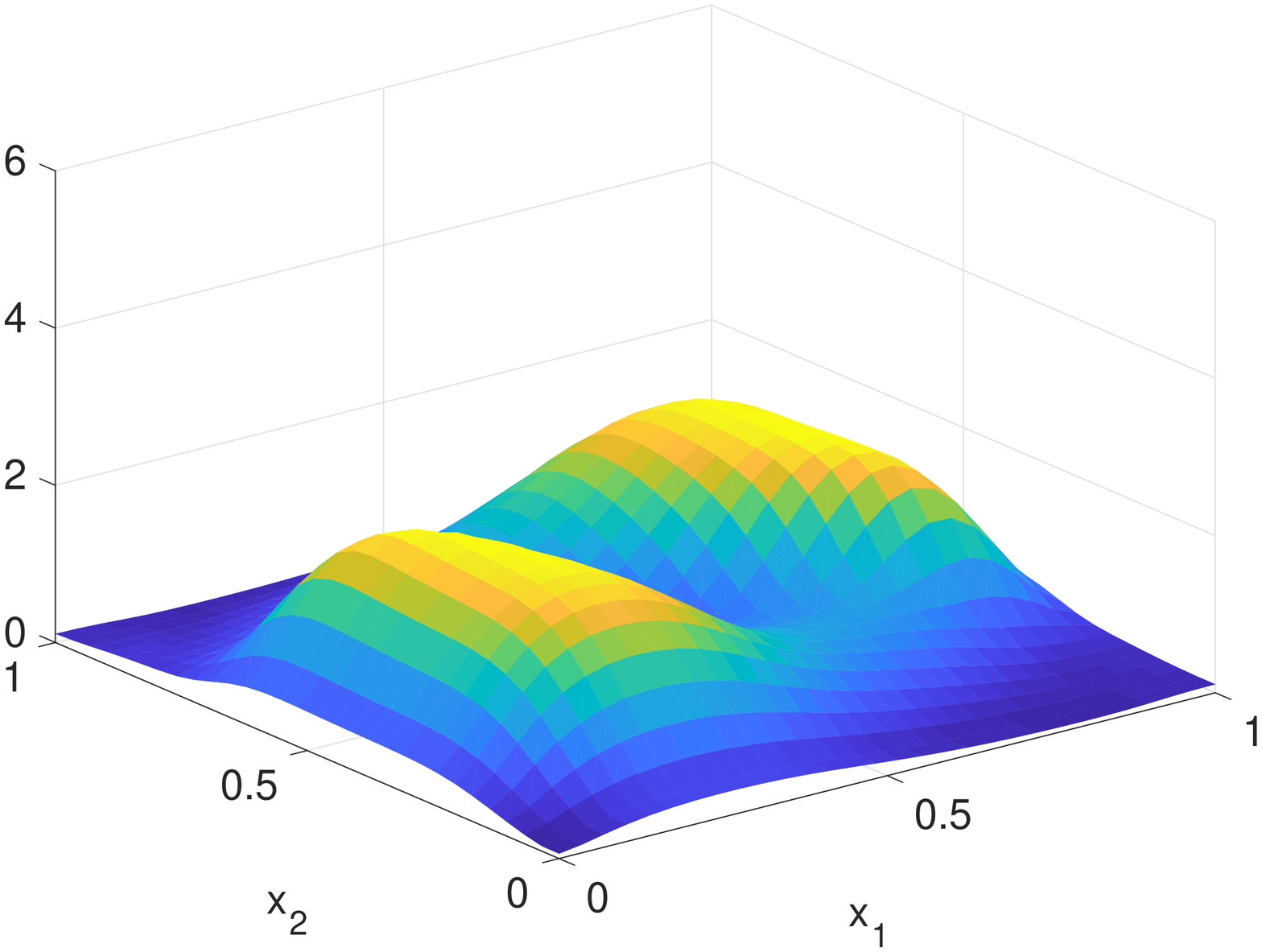}
    \end{subfigure}
    \begin{subfigure}[b]{0.22\textwidth}
        \centering
        \includegraphics[width=0.95\textwidth]{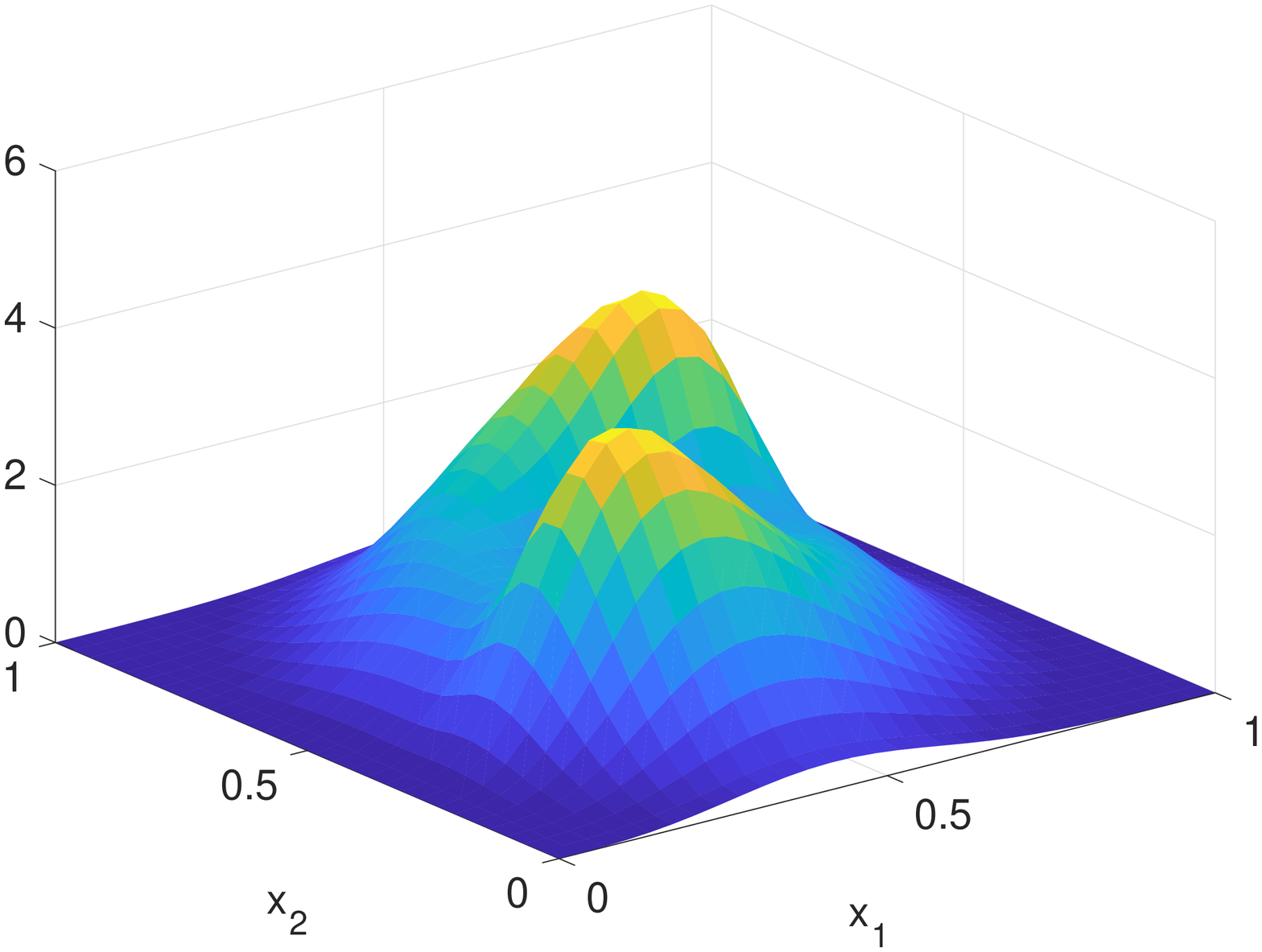}
    \end{subfigure}
    \begin{subfigure}[b]{0.22\textwidth}
        \centering
        \includegraphics[width=0.95\textwidth]{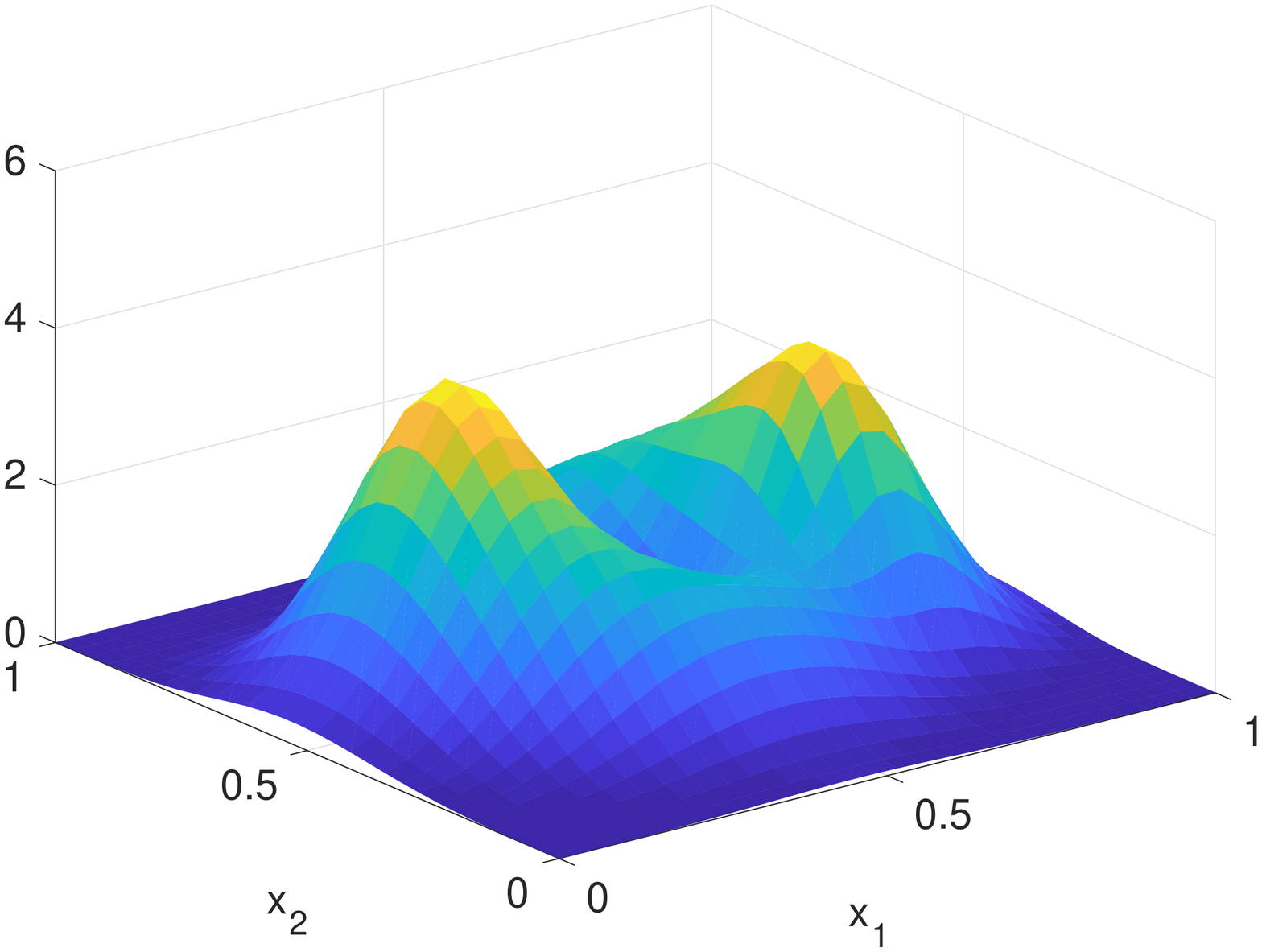}
    \end{subfigure}
    \caption{States of the large-scale agent systems (1st row), density estimates using KDE $p_{\text{KDE}}(x,t)$ (2nd row), outputs of the density filter $\Hat{p}(x,t)$ (3rd row) and the ground truth density $p(x,t)$ (4th row).}
    \label{fig:density filter}
\end{figure*}

\section{Conclusion} \label{section:conclusion}
We have presented a density filter for estimating the dynamic density of large-scale agent systems with known dynamics by combining KDE with infinite-dimensional Kalman filters. 
The density filter improved its estimation using the dynamics and the real-time states of the agents.
It was scalable to the population of agents and was proved to be convergent. 
All results held even if the agents' dynamics are nonlinear and time-varying. 
This algorithm can be used for many density-based optimization and control tasks of large-scale agent systems when density feedback information is required, and can be potentially extended to other dynamic density estimation problems when certain prior knowledge of its spatiotemporal evolution is available, such as population migration.
Our future work includes addressing the remaining problems noted in Remark \ref{remark:true covariance equation} and \ref{remark:uniform boundedness}, decentralizing the density filter and integrating it into density feedback control for large-scale agent systems.

\bibliographystyle{IEEEtran}
\bibliography{References}

\begin{thebibliography}{10}
\providecommand{\url}[1]{#1}
\csname url@samestyle\endcsname
\providecommand{\newblock}{\relax}
\providecommand{\bibinfo}[2]{#2}
\providecommand{\BIBentrySTDinterwordspacing}{\spaceskip=0pt\relax}
\providecommand{\BIBentryALTinterwordstretchfactor}{4}
\providecommand{\BIBentryALTinterwordspacing}{\spaceskip=\fontdimen2\font plus
\BIBentryALTinterwordstretchfactor\fontdimen3\font minus
  \fontdimen4\font\relax}
\providecommand{\BIBforeignlanguage}[2]{{%
\expandafter\ifx\csname l@#1\endcsname\relax
\typeout{** WARNING: IEEEtran.bst: No hyphenation pattern has been}%
\typeout{** loaded for the language `#1'. Using the pattern for}%
\typeout{** the default language instead.}%
\else
\language=\csname l@#1\endcsname
\fi
#2}}
\providecommand{\BIBdecl}{\relax}
\BIBdecl

\bibitem{foderaro2016distributed}
G.~Foderaro, P.~Zhu, H.~Wei, T.~A. Wettergren, and S.~Ferrari, ``Distributed
  optimal control of sensor networks for dynamic target tracking,'' \emph{IEEE
  Transactions on Control of Network Systems}, vol.~5, no.~1, pp. 142--153,
  2016.

\bibitem{silverman1986density}
B.~W. Silverman, \emph{Density Estimation for Statistics and Data
  Analysis}.\hskip 1em plus 0.5em minus 0.4em\relax CRC Press, 1986, vol.~26.

\bibitem{foderaro2012decentralized}
G.~Foderaro, S.~Ferrari, and M.~Zavlanos, ``A decentralized kernel density
  estimation approach to distributed robot path planning,'' in
  \emph{Proceedings of the Neural Information Processing Systems Conference,
  Lake Tahoe, NV}, 2012.

\bibitem{krishnan2018distributed}
V.~Krishnan and S.~Mart{\'\i}nez, ``Distributed control for spatial
  self-organization of multi-agent swarms,'' \emph{SIAM Journal on Control and
  Optimization}, vol.~56, no.~5, pp. 3642--3667, 2018.

\bibitem{de2018optimal}
M.~H. de~Badyn, U.~Eren, B.~A{\c{c}}ikme{\c{s}}e, and M.~Mesbahi, ``Optimal
  mass transport and kernel density estimation for state-dependent networked
  dynamic systems,'' in \emph{2018 IEEE Conference on Decision and Control
  (CDC)}.\hskip 1em plus 0.5em minus 0.4em\relax IEEE, 2018, pp. 1225--1230.

\bibitem{sain2002multivariate}
S.~R. Sain, ``Multivariate locally adaptive density estimation,''
  \emph{Computational Statistics \& Data Analysis}, vol.~39, no.~2, pp.
  165--186, 2002.

\bibitem{sheather1991reliable}
S.~J. Sheather and M.~C. Jones, ``A reliable data-based bandwidth selection
  method for kernel density estimation,'' \emph{Journal of the Royal
  Statistical Society: Series B (Methodological)}, vol.~53, no.~3, pp.
  683--690, 1991.

\bibitem{qahtan2016kde}
A.~Qahtan, S.~Wang, and X.~Zhang, ``Kde-track: An efficient dynamic density
  estimator for data streams,'' \emph{IEEE Transactions on Knowledge and Data
  Engineering}, vol.~29, no.~3, pp. 642--655, 2016.

\bibitem{julier2004unscented}
S.~J. Julier and J.~K. Uhlmann, ``Unscented filtering and nonlinear
  estimation,'' \emph{Proceedings of the IEEE}, vol.~92, no.~3, pp. 401--422,
  2004.

\bibitem{arasaratnam2009cubature}
I.~Arasaratnam and S.~Haykin, ``Cubature kalman filters,'' \emph{IEEE
  Transactions on automatic control}, vol.~54, no.~6, pp. 1254--1269, 2009.

\bibitem{chen2003bayesian}
Z.~Chen \emph{et~al.}, ``Bayesian filtering: From kalman filters to particle
  filters, and beyond,'' \emph{Statistics}, vol. 182, no.~1, pp. 1--69, 2003.

\bibitem{olfati2007distributed}
R.~Olfati-Saber, ``Distributed kalman filtering for sensor networks,'' in
  \emph{2007 46th IEEE Conference on Decision and Control}.\hskip 1em plus
  0.5em minus 0.4em\relax IEEE, 2007, pp. 5492--5498.

\bibitem{bandyopadhyay2014distributed}
S.~Bandyopadhyay and S.-J. Chung, ``Distributed estimation using bayesian
  consensus filtering,'' in \emph{2014 American control conference}.\hskip 1em
  plus 0.5em minus 0.4em\relax IEEE, 2014, pp. 634--641.

\bibitem{battistelli2014kullback}
G.~Battistelli and L.~Chisci, ``Kullback--leibler average, consensus on
  probability densities, and distributed state estimation with guaranteed
  stability,'' \emph{Automatica}, vol.~50, no.~3, pp. 707--718, 2014.

\bibitem{kalman1961new}
R.~E. Kalman and R.~S. Bucy, ``New results in linear filtering and prediction
  theory,'' \emph{Journal of basic engineering}, vol.~83, no.~1, pp. 95--108,
  1961.

\bibitem{falb1967infinite}
P.~L. Falb, ``Infinite-dimensional filtering: The kalman-bucy filter in hilbert
  space,'' \emph{Information and Control}, vol.~11, no. 1-2, pp. 102--137,
  1967.

\bibitem{bensoussan1971filtrage}
A.~Bensoussan, \emph{Filtrage optimal des syst{\`e}mes lin{\'e}aires}.\hskip
  1em plus 0.5em minus 0.4em\relax Dunod, 1971.

\bibitem{da2014stochastic}
G.~Da~Prato and J.~Zabczyk, \emph{Stochastic equations in infinite
  dimensions}.\hskip 1em plus 0.5em minus 0.4em\relax Cambridge university
  press, 2014.

\bibitem{sontag1995characterizations}
E.~D. Sontag and Y.~Wang, ``On characterizations of the input-to-state
  stability property,'' \emph{Systems \& Control Letters}, vol.~24, no.~5, pp.
  351--359, 1995.

\bibitem{dashkovskiy2013input}
S.~Dashkovskiy and A.~Mironchenko, ``Input-to-state stability of
  infinite-dimensional control systems,'' \emph{Mathematics of Control,
  Signals, and Systems}, vol.~25, no.~1, pp. 1--35, 2013.

\bibitem{pavliotis2014stochastic}
G.~A. Pavliotis, \emph{Stochastic processes and applications: diffusion
  processes, the Fokker-Planck and Langevin equations}.\hskip 1em plus 0.5em
  minus 0.4em\relax Springer, 2014, vol.~60.

\bibitem{cacoullos1966estimation}
T.~Cacoullos, ``Estimation of a multivariate density,'' \emph{Annals of the
  Institute of Statistical Mathematics}, vol.~18, no.~1, pp. 179--189, 1966.

\bibitem{chang1970practical}
J.~Chang and G.~Cooper, ``A practical difference scheme for fokker-planck
  equations,'' \emph{Journal of Computational Physics}, vol.~6, no.~1, pp.
  1--16, 1970.

\end{thebibliography}

\end{document}